\documentclass[a4 paper,11pt]{amsart}

\usepackage{amsmath,amssymb,amscd,amsthm} 
\usepackage{mathrsfs,euscript,ulem,mathabx}
\usepackage[alphabetic]{amsrefs}

\usepackage[dvipdfmx]{graphicx} 
\usepackage{hyperref}
\usepackage{xcolor}
\definecolor{unbleu}{rgb}{0.03, 0.15, 0.4}
\hypersetup{
pdfborder = {0 0 0},
colorlinks,
linkcolor=unbleu,
citecolor=unbleu,
urlcolor=unbleu
}

\usepackage{fancyhdr}

\usepackage{scalerel}

\newtheorem{theorem}{Theorem}[section]
\newtheorem{lemma}[theorem]{Lemma}
\newtheorem{proposition}[theorem]{Proposition}

\theoremstyle{definition}
\newtheorem{definition}[theorem]{Definition}

\newtheorem{example}[theorem]{Example}

\newcommand{\R}{\mathbb R}
\newcommand{\Z}{\mathbb Z}
\newcommand{\N}{\mathbb N}
\newcommand{\dd}{{\,\mathrm d}} 
\newcommand{\e}{\operatorname{e}} 

\newcommand{\boldi}{\boldsymbol{i}}
\newcommand{\boldj}{\boldsymbol{j}}
\newcommand{\boldk}{\boldsymbol{k}}

\newcommand{\red}[1]{\textcolor{red}{#1}}

\AtBeginDocument{%
   \def\MR#1{}
}

\begin{document}

\title[]{{\small On the absence of zero-temperature limit of equilibrium states for finite-range interactions on the lattice $\Z^2$}}

\author[J. R. Chazottes]{Jean-Ren\'e Chazottes}

\address{Centre de Physique Th\'eorique, CNRS, Institut Polytechnique de Paris, France}

\email{jeanrene@cpht.polytechnique.fr}

\author[M. Shinoda]{Mao Shinoda}

\address{Department of Mathematics, Ochanomizu University, 2-1-1 Otsuka, Bunkyo-ku, Tokyo, 112-8610, Japan}

\email{shinoda.mao@ocha.ac.jp}

\thanks{We thank Mike Hochman for having drawn our attention to the work of Durand \textit{et al.} \cite{DRS12}.}

\date{\today}

\begin{abstract}
We construct finite-range interactions on $\mathcal{S}^{\Z^2}$, where $\mathcal{S}$ is a finite set, for which the associated 
equilibrium states (i.e., the shift-invariant Gibbs states) fail to converge as temperature goes to zero. More precisely,
if we pick any one-parameter family $(\mu_\beta)_{\beta>0}$ in which $\mu_\beta$ is an equilibrium state at inverse temperature $
\beta$ for this interaction, then $\lim_{\beta\to\infty}\mu_\beta$ does not exist.  This settles a question posed by the first author and 
Hochman who obtained such a non-convergence behavior when $d\geq 3$, $d$ being the dimension of the lattice.

\smallskip
\noindent \textbf{Keywords.} lattice systems, equilibrium states, Gibbs measures, thermodynamic formalism, multidimensional subshift of finite type, Wang tilings, effective subshifts, Turing machine.
\end{abstract}

\maketitle

\tableofcontents

\section{Introduction and main result}

\subsection{Setting}

Leaving precise definitions till the next section, we work in the context of Gibbs measures for
shift-invariant absolutely summable interactions on a space of configurations of the form
$\mathcal{S}^{\Z^d}$, where $\mathcal{S}$ is a finite set. As explained below, we will in fact deal with equilibrium states, that is,
shift-invariant Gibbs measures. The problem we consider is the following. Given an interaction $\Phi$ and an inverse temperature $\beta>0$, there is a simplex of equilibrium states $\mathcal{E\!S}(\beta)$ associated with $\beta\Phi$ (which might not be a singleton for large values of $\beta$, as for instance in the Ising model). 
We ask the following:
\begin{quote}
What is the behavior of $\mathcal{E\!S}(\beta)$ when $\beta\to+\infty$?
\end{quote}
When there is a single equilibrium state $\mu_\beta$ for each $\beta$, this question is simply: Does the limit of $(\mu_{\beta})_{\beta>0}$ exist?
If it does, what is the limiting measure? (The natural topology in this problem is the weak topology, see below.) This question is 
connected with ground states. We need not explain what they are because they will play no explicit role in the present paper. 
Let us only say that a ground state for an interaction $\Phi$ is a probability measure supported on a certain closed subset of $\mathcal{S}^{\Z^d}$, 
possibly uncountable, which is invariant under the shift action (that is, a subshift), and determined by the ``maximizing configurations''
of the Hamiltonian of $\Phi$. We refer to \cite{GT2015} for details. (Notice that in this paper we use the convention sign followed in dynamical 
systems and also in \cite{Ruelle2004}, namely we ``prefer'' to maximize instead of minimizing.)

\subsection{Known results and main theorem}

If the answer to the above question is known in a number of particular examples, notably in relation with phase transitions,
see, \textit{e.g.}, \cite{DB1985,vE-F-S,Georgii}, the general study of this problem is pretty recent, and it was started by people working 
in ergodic theory and dynamical systems. They considered `potentials' on $\mathcal{S}^{\N}$ (or $\mathcal{S}^{\Z}$) for which there is a single 
equilibrium state, which is also a Gibbs measure, for each inverse temperature \cite{Bow75}. \footnote{A caution on the terminology is in order. In 
statistical physics, an interaction or a potential is a family of real-valued functions on $\mathcal{S}^{\Z^d}$ indexed by the finite subsets of $\Z^d$. 
For equilibrium states, a function deriving from the potential appears naturally and can be interpreted as the `mean energy per site'. In dynamical 
systems, people consider $d=1$ (the shift representing time evolution), and they only consider this mean energy per site that they call a potential.}
In a nutshell, the situation is the following. For locally constant `potentials', which correspond to finite-range interactions, convergence always takes 
place, and it is possible to describe the limit measures \cite{Bre03, CGU11, Lep05}. For Lipschitz `potentials', which correspond for instance to 
exponentially decaying pair interactions (as a function of the distance between sites), there is a rather surprising negative result. 
Recall that, in this class, equilibrium states and Gibbs measures coincide, and for a given potential and for each $\beta$ there is exactly one Gibbs 
measure. It was first proved in \cite{CH10} that there do exist potentials (or interactions) in this class such that the limit of
$(\mu_{\beta})_{\beta>0}$ does not exist when $\beta\to+\infty$.
Another construction was given in \cite{CR15}.

What happens when $d\geq 2$? In sharp contrast with the case $d=1$, it was proved in \cite{CH10} that, when $d\geq 3$, then one can construct
a finite-range interaction (with $\mathcal{S}=\{0,1\}$) such that, for \textit{any} family $(\mu_\beta)_{\beta>0}$ in which $\mu_\beta$ is an 
equilibrium state for this interaction at inverse temperature $\beta$, the limit $\lim_{\beta\to\infty}\mu_\beta$ does not exist. (We will comment below
on this rather subtle statement.)
The case $d=2$ is left as an open problem in \cite{CH10} (for reasons that will be explained later on), and in this paper, we solve it. 
More precisely, our main theorem is the following.
\begin{theorem}[Main theorem] 
\label{maintheorem}
\leavevmode\\
There exists a finite set $\mathcal{S}$ and a finite-range interaction on $\mathcal{S}^{\Z^2}$, such that for any one-parameter family $(\mu_\beta)_{\beta>0}$ in 
which $\mu_\beta$ is an equilibrium state (\textit{i.e.}, a shift-invariant Gibbs measure) at inverse temperature $\beta$, the limit $\lim_{\beta\to\infty}
\mu_\beta$ does not exist.
\end{theorem} 

Several comments are in order. First, if there were a unique equilibrium state/Gibbs measure for each $\beta$, then there would be a unique choice 
for $\mu_\beta$, and the previous result could be formulated more transparently: there exist finite range interactions such that the limit $
\lim_{\beta\to\infty}\mu_\beta$ does not exist. But in our example we didn't look if uniqueness holds at low temperature.
Second, by compactness (in the weak topology) of the space of probability measures, if we take any sequence $(\beta_\ell)_{\ell\geq 1}$ of
inverse temperatures such that $\beta_\ell\to+\infty$, there exists a subsequence $(\ell_i)$ such that the sequence
$(\mu_{\beta_{\ell_i}})_{i\geq 1}$, in which $\mu_{\beta_{\ell_i}}$ is an equilibrium state, has a limit.
Our result, as well as the one in \cite{CH10} mentioned above, is about continuous-parameter families. Third, and last, there is
nothing new in the fact that one can choose \textit{some} divergent family of equilibrium states. Consider for instance
the nearest-neighbor Ising model in which one can choose a family which alternates, when $\beta$ is large enough, between the $+$ and $-$ phases. However, it is always possible to choose families which converge to either $\delta_-$ or $\delta_+$. In our example, \textit{it is not
possible} to choose \textit{any} family which converges to a ground state. 
Let us also mention that in \cite{CR15} such a non-convergence result (for any $d\geq 2$) was obtained, but for non-locally constant Lipschitz `potentials'. 

\subsection{More comments}
The fact that the Gibbs measures of an interaction can behave in a `chaotic' way when temperature goes to zero seems to have been first 
proved in \cite{vE-Wioletta} for a class of examples of nearest-neighbor, bounded-spin models, in any dimension. In that example, $\mathcal{S}$
is the unit circle. The paper \cite{CH10} was the first to exhibit this kind of behavior for models with a finite number of `spin' values at each site. 
In the above mentioned paper \cite{CR15}, a stronger property is studied namely `sensitive dependence'.
Roughly, it means that the non-convergence can indeed occur along any prescribed sequence of temperatures going to zero, by making an arbitrarily small perturbation of the original interaction. We believe that our example exhibits this property but we did not try to prove it.
Finally, let us mention that we only deal with equilibrium states, that is, shift-invariant Gibbs measures. It is well known that there can exist non-shift invariant Gibbs measures at low temperature, \textit{e.g.}, in the three-dimensional Ising model where the so-called `Dobrushin states' appear
\cite{dobrushin-Ising3D,DB1985}. The situation is unclear in that case.

\subsection{On the role of symbolic dynamics}
In \cite{CH10}, \cite{CR15} and the present work, a central role is played by symbolic dynamics, in particular the construction of subshifts with certain 
properties.
Informally, a subshift is a subset of configurations in $\mathcal{S}^{\Z^d}$ defined by a (finite or infinite) set of `patterns' which cannot appear anywhere in 
these configurations. When $d=1$, we say that we have a 1D subshift. A prominent class of 1D subshifts is that of subshifts of finite type for which there 
are finitely many 
forbidden patterns. They play a central role to `encode' certain differential dynamical systems such as Axiom A diffeomorphisms \cite{Bow75}. 
There are many 1D subshifts that are not of finite type which were introduced for various purposes (for instance the Thue-Morse subshift defined by 
substitution rules); see for instance \cite{kurka-book}.

There is a striking and dramatic difference between 1D and 2D subshifts of finite type. For instance, it is formally undecidable whether a 2D subshift of finite type 
is empty or not. This undecidability problem is closely related to the existence of nonempty shifts of finite type without periodic points, or, equivalently, the
existence of Wang tile sets (their definition is given below) such that one can tile the plane but never in a periodic fashion \cite{LS02}.

1D subshifts of finite type are closely related to the zero temperature limit of (one-dimensional) Gibbs measures of finite-range potentials: the limiting 
measure, which always exists, is necessarily supported on a subshift of finite type. The above  mentioned examples of non-convergence for non-finite-
range potentials \cite{CH10,CR15} rely on the construction of some subshifts which are necessarily not of finite type. Roughly speaking, the idea is to cook 
up two subshifts of $\mathcal{S}^{\Z}$, each carrying only one shift-invariant probability measure (among other properties), and a (non-finite-range) potential 
such that the corresponding one-parameter family of Gibbs measures $(\mu_\beta)_{\beta>0}$ accumulates at the same time on the two measures as $
\beta\to+\infty$. 

In dimension higher than one, we previously said that this non-convergence phenomenon can arise for finite-range potentials. 
The underlying phenomenon which we exploit is that one can imbed (in a way precised below) any (effective) 1D subshift into a higher-dimensional
subshift of finite type. In \cite{Hoc09}, there is a construction which allows to imbed a 1D effective subshift into a 3D subshift of finite type, which is the one 
used in \cite{CH10}. In this paper we use another construction from \cite{DRS12} based on `hierarchical self-simulating tilings'. It permits to imbed any (effective) 1D subshift into a 2D subshift of 
finite type. This is a rather cumbersome construction (that we will partly describe it below), although the underlying ideas are 
simple. (Let us mention that a different embedding construction is given in \cite{AS13}.) Moreover, we use the construction of 
certain 1D subshifts given in \cite{CR15}. It is somewhat more flexible than the one used in \cite{CH10}. Once we have a 2D subshift of finite type built up from a certain 1D subshift, 
we can then define a finite-range potential which `penalizes' the forbidden patterns (which are finitely many).

\subsection{Organization of the paper}
In Section \ref{Preliminaries} we set the necessary definitions and notations for equilibrium states, subshifts and Wang tilings.
In Section \ref{Imbedding-Proposition} we state the embedding theorem of Durand et al. \cite{DRS12} and establish a key proposition which results from their
construction. In particular, we explain some of the ideas of the proof of the embedding theorem.
Next, we construct in Section \ref{Construction-of-a-base-system} a certain 1D effective subshift that serves as a `base' for a 2D subshift of finite type,
and we define an associated finite-range interaction. Section \ref{Estimates-on-admissible-patterns} contains some estimates involving the admissible 
patterns of the 2D subshift of finite type.
Finally, we prove in Section \ref{Non-convergence} our main result (Theorem \ref{maintheorem}), namely that that for every one-parameter family
$(\mu_\beta)_{\beta>0}$ in which $\mu_\beta$ is an equilibrium state at inverse temperature $\beta$ for the above interaction, $\lim_{\beta\to\infty}\mu_\beta$ does 
not exist.

\section{Equilibrium states, subshifts and tilings}
\label{Preliminaries}

The configuration space is $\mathcal{S}^{\Z^d}$, where $\mathcal{S}$ is a finite set and $d\geq 1$ is an integer. 
Regarding equilibrium states, we are interested in $d=2$, but we will also consider the case $d=1$ to
construct some subshifts needed in the proof of the main result.
On $\mathcal{S}^{\Z^d}$, we have the shift operator $\sigma$ defined by
\[
\sigma^{\boldi}(x)_{\boldj}=x_{\boldi+\boldj}
\]
where $x=(x_{\boldk})_{\boldk\in\Z^d}\in\mathcal{S}^{\Z^d}$, $\boldi,\boldj\in \Z^d$. In the language of symbolic dynamics \cite{LS02},
$(\mathcal{S}^{\Z^d},\sigma)$ is the $d$-dimensional full shift over $\mathcal{S}$. As usual, $\mathcal{S}^{\Z^d}$ is given the product topology, 
which is generated by the cylinder sets, and thus it is a compact metrizable space.
We denote by $\mathfrak{B}$ the Borel $\sigma$-algebra which coincides with the $\sigma$-algebra generated by cylinder sets.

\subsection{Equilibrium states}
We only recall a few definitions and facts, mainly to set notations. We refer to \cite{Georgii,Ruelle2004} for details, as well as to \cite{Kel98} for
a viewpoint from ergodic theory for $\Z^d$-actions.
The basic ingredient is a shift-invariant summable
interaction $\Phi=(\Phi_\Lambda)_{\Lambda \Subset \Z^d}$. ($\Lambda \Subset \Z^d$ means that $\Lambda$ is a nonempty finite subset of $\Z^d$.)
More precisely, for each $\Lambda\Subset \Z^d$, $\Phi_\Lambda:\mathcal{S}^{\Z^d}\to\R$ is $\mathfrak{B}_\Lambda$-measurable,\footnote{\,$\mathfrak{B}_\Lambda$ is the $\sigma$-algebra generated by the coordinate maps $\omega\mapsto \omega_x$ when $x$ is restricted to $\Lambda$.}
$\Phi_\Lambda(x)=\Phi_\Lambda(\tilde{x})$ whenever $x$ and $\tilde{x}$ coincide on $\Lambda$, $\Phi_{\Lambda+\boldi}
=\Phi_\Lambda\circ \sigma^{\boldi}$ for all $\boldi\in\Z^d$, and $\sum_{\Lambda\ni\, 0} \|\Phi_\Lambda\|_\infty<\infty$. 
We say that $\Phi$ is of finite range if there exists $R>0$ such that $\Phi_\Lambda\equiv 0$ whenever $\text{diam}(\Lambda)>R$.
Given $\Phi$, define the function $\phi:\mathcal{S}^{\Z^d}\to\R$ by
\begin{equation}\label{def-phi-from-Phi}
\phi(x)=\sum_{\substack{\Lambda\ni\, 0 \\ \Lambda\Subset\Z^d}} \frac{\Phi_\Lambda(x)}{|\Lambda|}\,.
\end{equation}
By definition, the equilibrium states of $\beta\Phi$ are the shift-invariant probability measures which maximize the 
quantity
\[
\int \beta \phi\dd\nu + h(\nu)
\]
over all shift-invariant probability measures $\nu$ on $\mathcal{S}^{\Z^d}$. Here $h(\nu)$ is the entropy of $\nu$ (also called the mean entropy per 
site in statistical physics), and the supremum (which is attained) is called the (topological) pressure and is
denoted by $P(\beta\phi)=P(\sigma,\beta\phi)$. 
For the class of interactions we consider, shift-invariant Gibbs measures coincide with equilibrium states (see \cite[Chapter 15]{Georgii} or \
\cite[Theorem 4.2]{Ruelle2004}).
We use the terminology and convention of dynamical systems and thermodynamic formalism where one  ``prefers'' to maximize,
whereas in statistical physics one ``prefers'' to minimize.
Note that if $\Phi$ is of finite range then $\phi$ is locally constant, which means that the values of $\phi(x)$ are determined by
finitely many coordinates of $x$.

\subsection{Subshifts, 2D subshifts of finite type and Wang tilings}

We refer to \cite{LS02} for more details.

We now turn to symbolic dynamics.
For $\Lambda\Subset\Z^d$ and $x\in \mathcal{S}^{\Z^d}$, $x_{\Lambda}$ is
the restriction of the configuration $x$ to $\Lambda$. An element $\omega\in \mathcal{S}^\Lambda$ is called a $\Lambda$-pattern, or simply a 
pattern, and $\Lambda$ is its support. But only the ``shape'' of $\Lambda$ matters, not its ``location'' in $\Z^d$.
More precisely, define the equivalence relation $\sim$ on the set of finite subsets of $\Z^d$ by setting
$\Lambda\sim \Lambda'$ if and only if there exists $\boldi\in\Z^d$ such that $\Lambda'=\Lambda+\boldi$. Hence, a support of a pattern
is an equivalence class for $\sim$. Let us also denote by $\sim$ the equivalence relation saying that two patterns
$\omega\in \mathcal{S}^\Lambda$ and $\omega'\in \mathcal{S}^{\Lambda'}$ are congruent if there exists $\boldj\in\Z^d$ such
that $\Lambda'=\Lambda+\boldj$ and $\omega'_{\boldi+\boldj}=\omega_{\boldi}$ for all $\boldi\in \Lambda$.
In the sequel, for the sake of simplicity, we will several times consider a ``localized'' $\Lambda$, and by $\omega\in\mathcal{S}^\Lambda$ we 
will mean any pattern $\omega'\sim\omega$.

For $n\geq1$ let
\[
\Lambda_n=\{-n+1, \ldots, 0,  \ldots, n-1\}^d
\]
which is the discrete $d$-dimensional cube with volume $\lambda_n=|\Lambda_n|=(2n-1)^d$.
When $d=1$, patterns of the form $\omega_0\cdots \omega_{n-1}$, where $\omega_i\in \mathcal{S}$ and $n\geq 0$, are called $n$-strings or simply strings.
Given a $\Lambda$-pattern $\omega$, let
\[
[\omega]=\big\{x\in \mathcal{S}^{\Z^d}: x_{\Lambda}=\omega\big\}
\]
denote the corresponding cylinder set. Given a finite set of patterns $P$, we write $[P]=\bigcup_{p\,\in P} [\,p\hspace{.5pt}]$.

A (nonempty) subset $X$ of $\mathcal{S}^{\Z^d}$ is a {\it subshift} if it is closed and $\sigma$-invariant.
Equivalently, $X\subseteq\mathcal{S}^{\Z^d}$ is a subshift if there exists a set $F$ of patterns such that
$X=X_F$ where
\[
X_F=\left\{x\in \mathcal{S}^{\Z^d}: \text{no pattern from}\ F\ \text{appears in}\ x\right\}.
\]

Thus $F$ is the set of ``forbidden'' patterns.
Note that $X_F$ may be empty and that different forbidden sets may generate the same subshift.
If $F$ is empty, $X_F=\mathcal{S}^{\Z^d}$.
A subshift $X$ is a \textit{of finite type} if there exists a finite set $F$ such that $X=X_F$.
We will use the abbreviation SFT for ``subshift of finite type''.
A subshift $X$ is \textit{effective} if there exists a recursively enumerable set $F$ such that $X=X_F$, that is, if we can have a 
Turing machine which, given no input, lists out the elements of $F$. Let us remark that the class of effective subshifts is countable, so we 
apparently rule out ``most'' subshifts, but all known examples are in this class, provided that they are defined using computable parameters, 
which is not a restriction in practice.

Given a subshift $X$ and an integer $n\geq1$, define the set of (locally) admissible $\Lambda_n$-patterns by
\[
\mathcal{P}_{X,n}=\left\{ \omega \in \mathcal{S}^{\Lambda_n}: \text{no forbidden pattern of}\  X\ \text{appears in}\ \omega\right\}.
\]
Finally, the set of probability measures on $\mathcal{S}^{\Z^d}$ is given the weak topology. A sequence of probability measures
$(\mu_k)_{k\geq 1}$ converges to a probability measure $\mu$ if, for any cylinder set $B$, $\mu_k(B)\to\mu(B)$, as $k\to+\infty$.

Let us briefly explain how a two-dimensional SFT can be seen as a Wang tiling, and vice versa. 
Working with tilings is better adapted for some constructions we use later on.
We consider tiles which are unit squares with colored sides. The colors are taken from a finite set
$\mathscr{C}$. For visualization purposes, one can actually use colors, but it can be more convenient to use symbols or integers. Hence, the
word ``color'' means any element from a finite set of symbols $\mathscr{C}$.
Hence a tile is a quadruple of colors (left, right, top and bottom ones), i.e., an element
of $\mathscr{C}^4$. A tile set is a subset $\tau\subset \mathscr{C}^4$. A Wang tiling with tiles from
$\tau$ is a mapping $x:\mathbb{Z}^2\to\tau$ which respects the color matching condition: abutting edges of adjacent tiles 
must have the same color. We shall simply say that it is a $\tau$-tiling.
See Fig. \ref{fig:example-Wang-tiling} for an example (where the colors are not only put on edges to ease visualization).
We can naturally identify each such tiling with a point $x=(x_{i,j})_{(i,j)\in\Z^2}\in \tau^{\Z^2}$, interpreting $\tau$ as an alphabet. 
The set $W\subset \tau^{\Z^2}$ of all $\tau$-tilings is obviously a subshift of
finite type (called the Wang shift of $\tau$). Conversely, a SFT can be regarded as a Wang shift.
In this paper, a tiling will mean a Wang tiling.

\begin{figure}[htbp]
\begin{center}
\includegraphics[width=250pt]{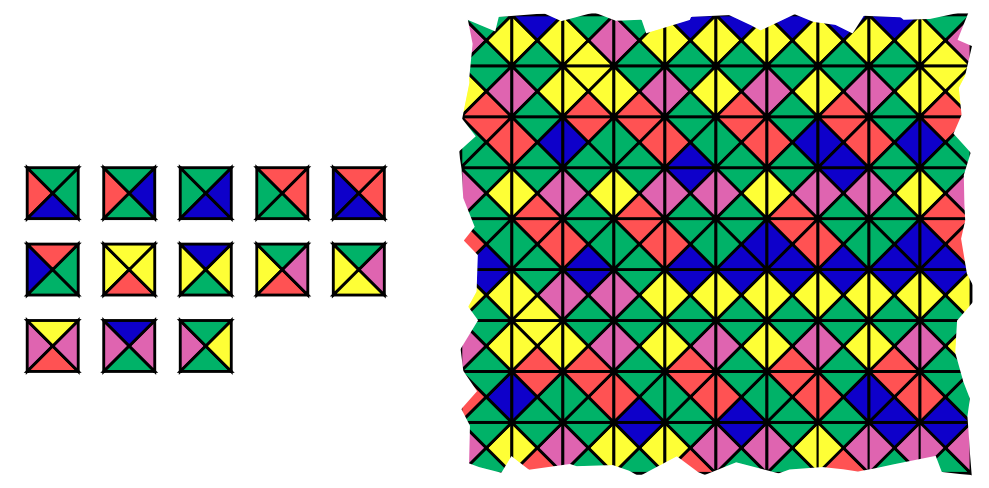}
\caption{An example of tile set $\tau$ and a region of $\Z^2$ legally tiled using this tile set.} \label{Wang tile}
\label{fig:example-Wang-tiling}
\end{center}
\end{figure}

\section{Imbedding a 1D effective subshift into a 2D subshift of finite type} \label{Imbedding-Proposition}

We are going to outline how to imbed a one-dimensional effective subshift into a two-dimensional SFT, as explained in detail in \cite{DRS12}. 
We define the ``vertical extension'' of a subshift $X\subset \mathcal{A}^\Z$ as
\[
\widehat{X}=\big\{ \hat{x}=(x_{i,j})_{(i,j)\in\Z^2}\in  \mathcal{A}^{\Z^2} : \forall j, (x_{i,j})_{i\in\Z}\in X
,\; \forall i,j, x_{i,j}=x_{i,j+1}\big\}.
\]
The following theorem and proposition are key-results in our construction. 
\begin{theorem}\cite[Theorem 10]{DRS12}\label{imbedding}
\leavevmode\\
Let $X$ be a one-dimensional effective subshift over a finite alphabet set $\mathcal{A}$. 
Then there exist finite alphabets $\mathcal{C}$ and $\breve{B}\subset \mathcal{A}\times \mathcal{C}$,
and a two-dimensional SFT $Y\subset \breve{B}^{\Z^2}$ such that $\pi(Y)=\widehat{X}$ where $\pi$
is the projection from $\breve{B}$ to $\mathcal{A}$.\footnote{We define in the obvious way the projection acting on
patterns or configurations by applying $\pi$ to every symbol.}
\end{theorem}
In plain words, every sequence from a 1D effective subshift can be obtained as a projection of a configuration of some 2D subshift of 
finite type in an extended alphabet.

Configurations in $Y$ can be seen as configurations from $\widehat{X}$ marked with some extra symbols taken from $\mathcal{C}$.
These symbols form a new ``layer" which is superimposed on top of $\widehat{X}$.
For reasons explained below, the layer where configurations of $\widehat{X}$ appear is called the {\it input layer}, whereas the layer where patterns over
$\mathcal{C}$ appear is called the {\it computation layer}. Then $\pi$ erases the superimposed layer of data corresponding to $\mathcal{C}$.

The following proposition will play an important role in some computations later on.
\begin{proposition} \label{computation layer}
\leavevmode\\
For each $n\geq 1$ there exists $c_n\leq (2n-1)^2$ such that
\[
|\pi^{-1}\, \widehat{\omega} \cap \mathcal{P}_{Y, n}| =c_n
\]
for all $\omega \in \mathcal{P}_{X,n}$, where $\widehat{\omega}\in\mathcal{P}_{\widehat{X},n}$.
\end{proposition}
This proposition  says that the number of admissible patterns in the computation layer does not depend on the input layer.
To prove it, we need some elements of the proof of Theorem \ref{imbedding} found in \cite{DRS12}.
The basic idea of the proof is to run a Turing machine $\EuScript{M}_X$ which checks the forbidden strings of $X$.
The transition rules of $\EuScript{M}_X$ are converted into tiling constraints, since they are described locally. 
Then a `space-time diagram' of $\EuScript{M}_X$ is (almost) a tiling based on the corresponding tile set. We can consider the horizontal
dimension as `space', given by the symbols on the tape of the Turing machine, whereas the vertical dimension is `time' which is
given by successive computations of the Turing machine.


There are two difficulties to check forbidden strings of $X$. First, we have to check arbitrary long input strings, since the length of forbidden strings may not be 
bounded. Second, we need to start the Turing machine at every site, since we should check every string starting at every site.
In order to overcome the first difficulty, we will consider a tile set which admits a hierarchical structure, the {\it self-simulating structure} defined in 
Subsection \ref{Self-simulating structure}.
The way to solve the second difficulty is explained in Subsection \ref{Simulation to check forbidden strings}.

Colloquially, the idea to organize the computations uses fixed-point self-similar tilings. The idea of a self-similar fixed-point tile set can be sketched as follows. 
It is well known that tilings can be used to simulate computations, in the sense that for any Turing machine one can construct a tile set simulating it:
use rows of tiles to simulate the tape in the machine, with successive rows corresponding to consecutive states of the machine.
In turn, these computations can be used to guarantee the desired behavior of bigger blocks, called 
macro-tiles. So, for a desired behavior of macro-tiles, we can construct a  tile set which guarantees this behavior. If these tiling rules coincide 
with the rules for macro-tiles, we get self-similarity as a consequence. The way to achieve this is to use an idea very close to the classical Kleene fixed-point theorem
in computability theory (and which was for instance used to construct self-reproducing automata by von Neumann).

\subsection{Self-simulating structure} \label{Self-simulating structure}
In order to get the self-simulating structure, we consider tilings with ``macro tiles".

\begin{definition}[Macro tiles]
Consider a tile set $\tau$ and an integer $N\geq1$.
A pattern $\omega$ over $\tau^{\{0, 1, \ldots, N-1\}^2}$ is called a {\it $\tau$-macro tile with zoom factor $N$},
if all tiles in $\omega$ satisfy the color matching property.
Denote by $\tau^{(N)}$ the set of all $\tau$-macro tiles with zoom factor $N$.
Every side of a macro-tile consists of a sequence of $N$ colors and we call it a macro-color. 
\end{definition} 
It is easy to see that a $\tau$-tiling can be seen as a $\tau^{(N)}$-tiling. 
In particular, we pay attention to the situation when a $\tau$-tiling can be split uniquely into macro-tiles which acts like tiles from another tile set $\rho$.

\begin{definition}[Simulation]
Let $\rho$ and $\tau$ be tile sets and $N\geq1$.
The tile set $\rho$ is {\it simulated} by a tile set $\tau$ with zoom factor $N$ if 
there exists an injective map $r: \rho\rightarrow \tau^{(N)}$ such that
\begin{itemize}
\item $t,s \in \rho$ satisfy the color matching property if and only if $r(t), r(s)$ satisfy the color matching property;
\item for every $\tau$-tiling there exists a unique vertical and horizontal $N\times N$ split such that every pattern in the $N\times N$ square is the image of an element in $\rho$ by $r$.
\end{itemize}

\end{definition}

\begin{example}[Coordinate tile]\cite{DRS12}\label{coordinate tile}
Consider a tile set $\rho$ whose element is colored by only one color, namely ``0'', and consider $\rho=\{0\}^4$.
We define a tile set which simulates $\rho$.
Let $N\geq2$ and $\mathscr{C}=(\Z/N\Z)^2$.
Define a tile set $\tau$ by
\[
\tau=\left\{t \in \mathscr{C}^4 : t_b=t_\ell=(i,j), t_r=(i+1, j), t_t=(i, j+1)\ (i,j)\in \mathscr{C}\right\}.
\]
Define a map $r: \rho \rightarrow \tau$ by 
$r((0,0,0,0))=$ the $\tau$-macro tile with zoom factor $N$ whose macro colors are
$(0,0) (1,0) \cdots (N-1,0)$ for the bottom and top,
$(0,0) (0,1) \cdots (0, N-1)$ for the left and right; see Fig. \ref{simulation macro tile}.
We call the tile $r((0,0,0,0))$ the {\it coordinate tile} with size $N$.
\end{example}

\begin{figure}[htbp]
\begin{center}
\includegraphics[width=300pt]{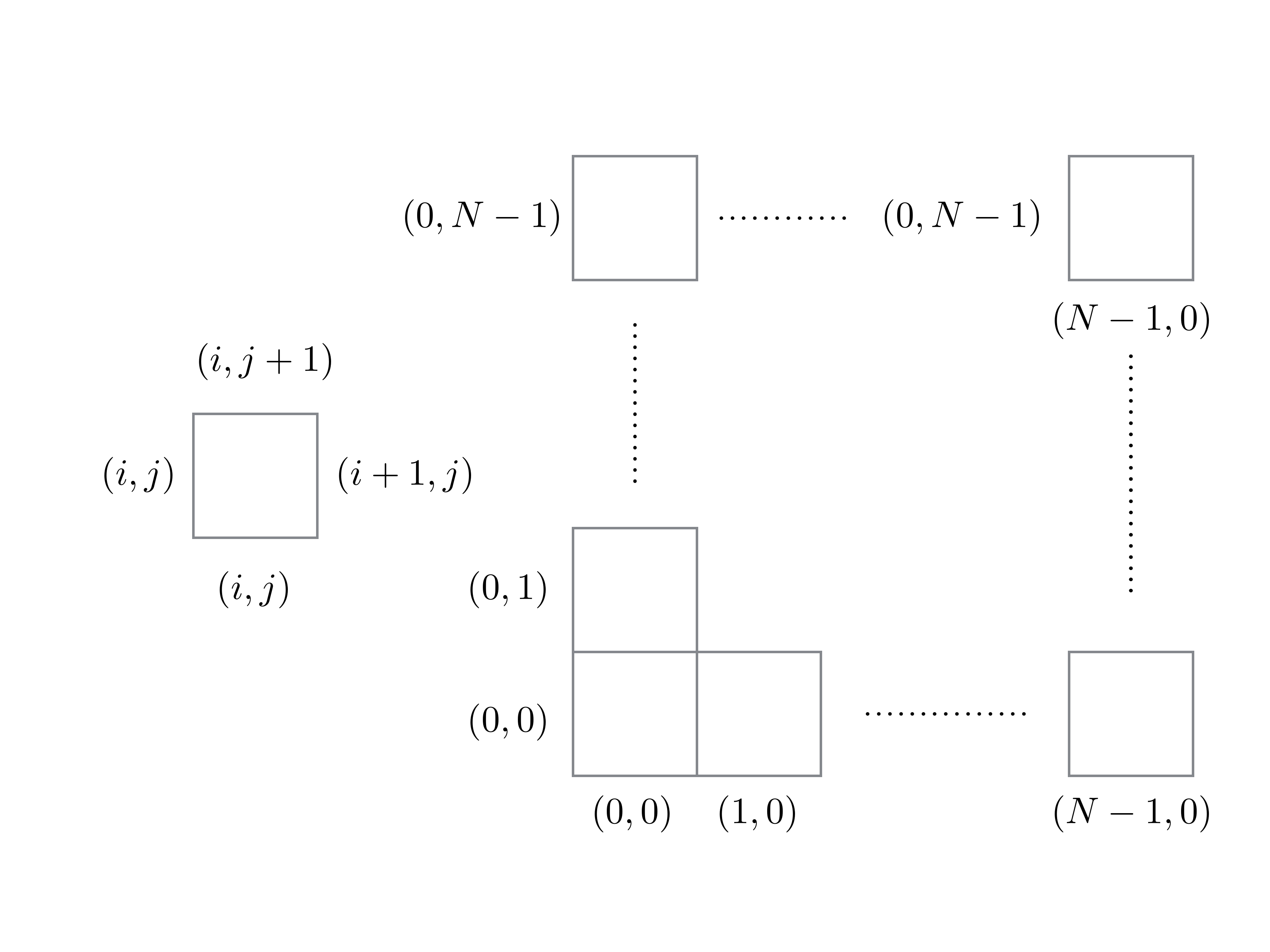}
\caption{A tile in $\tau$ (left) and a $\tau$ macro tile with zoom factor $N$ (right).} \label{simulation macro tile}
\end{center}
\end{figure} 

We will consider a sequence $\{\tau_k\}_{k\geq0}$ of tile sets with the following properties:
the level $k+1$ tile set $\tau_{k+1}$ is simulated by the level $k$ tile set $\tau_{k}$ with zoom factor $N_{k+1}$;
for every $k$ the tile set $\tau_k$ describes simulation of $\EuScript{M}_X$;
the zoom factors $N_k$ increase and macro tiles of $\tau_k$ can treat long input strings as $k$ increases. 
See Fig. \ref{sequence of macro tiles}.
 
We start with the construction of a tile set which simulates a $\ell$-bits colored tile set, $\rho \subset (\{0,1\}^{\ell})^4$. 
Consider a Turing machine $\EuScript{M}_\rho$ which checks whether a given four $\ell$-bits input represents a tile in $\rho$.
Superimposing other tiles on coordinate tiles (Example \ref{coordinate tile}), 
we define a tile set which simulates $\rho$.
The superimposed tiles make another ``layer" on the coordinate tiles. 
Let $c=(c_b, c_l, c_t, c_r)$ be a four $\ell$-bits input. 
Consider a macro tile with size $N$ consisting of  coordinate tiles with size $N$.
Tiles whose color contain $0$ are called boundary tiles. 
(In Fig.\ \ref{structure of macro tile} the grey  and green zones are the boundary tiles.)
On the middle $\ell$ tiles of the bottom (respectively left, top and right) side of the boundary, we distribute a $\ell$-bits color $c_b$ (resp.\ $c_l, c_t$ and $c_r$).
For the rest of the boundary tiles we distribute $0$.
Since we use the coordinate tiles, each tile ``knows" its coordinate and we can distribute colors like this.

In the middle square of the macro tile (the red and yellow zones in Fig.\ \ref{structure of macro tile})
we put tiles which describes a universal Turing machine with a program of $\EuScript{M}_\rho$.
Since the rule of a universal Turing machine is given by local rules, 
they can be embedded into tiles. 
Conveying the $\ell$-bits colors on the boundary to the middle square by wires, 
we let the universal Turing machine to know the input. 
Then a space-time diagram with the input $c=(c_b, c_l, c_t, c_r)$ appears in the square.
Since each tile ``knows" its coordinate, this structure is arranged easily. 
The size $N$ of macro tiles is chosen to be  large enough to contain these structures and to finish the simulation on the computation zone. 
If $c\in \rho$, the simulation doesn't halt and we have a macro tile.
If $c\notin \rho$, the simulation halts and there is no macro tile with this structure.
Since a universal Turing machine is deterministic, there is a one-to-one correspondence between an input and the pattern of the simulation.
Hence the tile set, say $\eta$, which makes this macro tile simulates $\rho$.

Note that the number of tiles in $\eta$ does not depend on $\EuScript{M}_\rho$
since we use a universal Turing machine.
Moreover, tiles for boundaries, wires and the space-time diagram do not depend on $\EuScript{M}_\rho$
and only those for the program do. 
By using this fact, we modify the program on a macro tile of $\eta$ and get self-simulating structure. 
We replace the program of $\EuScript{M}_\rho$ with the following three programs:
the program to make the boundaries, wires and computation structures; 
 the program of $\EuScript{M}_X$;
and the program which rewrites itself.
Then $\eta$ simulates a tile set whose macro tiles have the same structure as in Figure 3 and carries the same program as in the macro tiles of $\eta$.

\begin{figure}[htbp]
\begin{center}
\includegraphics[width=300pt]{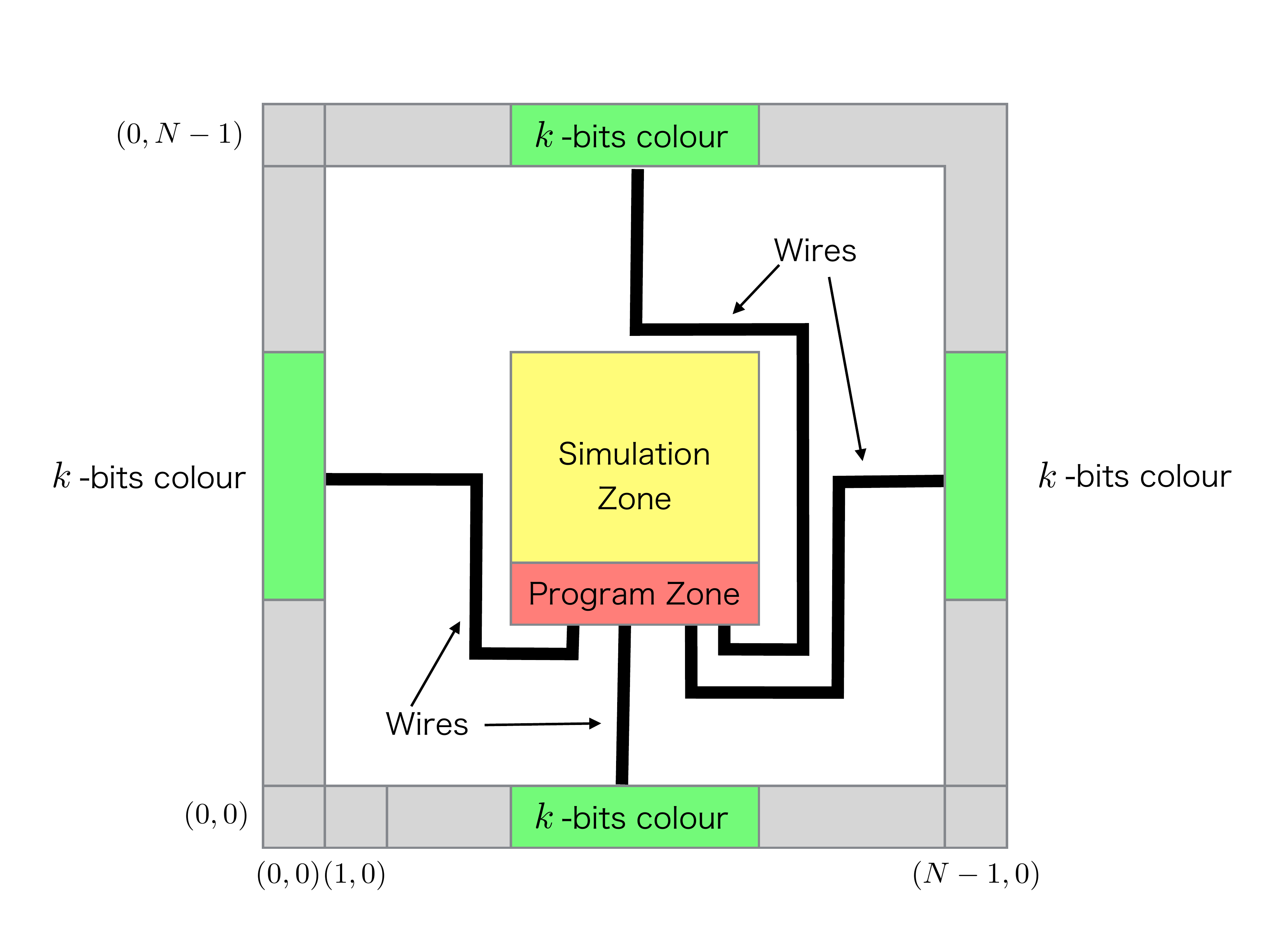}
\caption{A macro tile simulating a tile colored by $k$-bits.} \label{structure of macro tile}
\end{center}
\end{figure} 

Using this construction, we can make a sequence of tile sets which simulate tiles in the next level and carry the same program. 
The first level tile set $\tau_0$ simulates the second one $\tau_1$ with zoom factor $N_0$, then
$\tau_1$ simulates the third one $\tau_2$ with zoom factor $N_1$, and so on and so forth.
See Fig. \ref{sequence of macro tiles}.
Since $\tau_0$ simulates $\tau_1$ and $\tau_1$ simulates $\tau_2$, 
the patterns of $\tau_0$ in $N_0\times N_0$ squares must represent tiles in $\tau_2$.
If we zoom out, we can see $\tau_2$-tiling in a $\tau_0$-tiling and so on. 

\begin{figure}[htbp]
\begin{center}
\includegraphics[width=300pt]{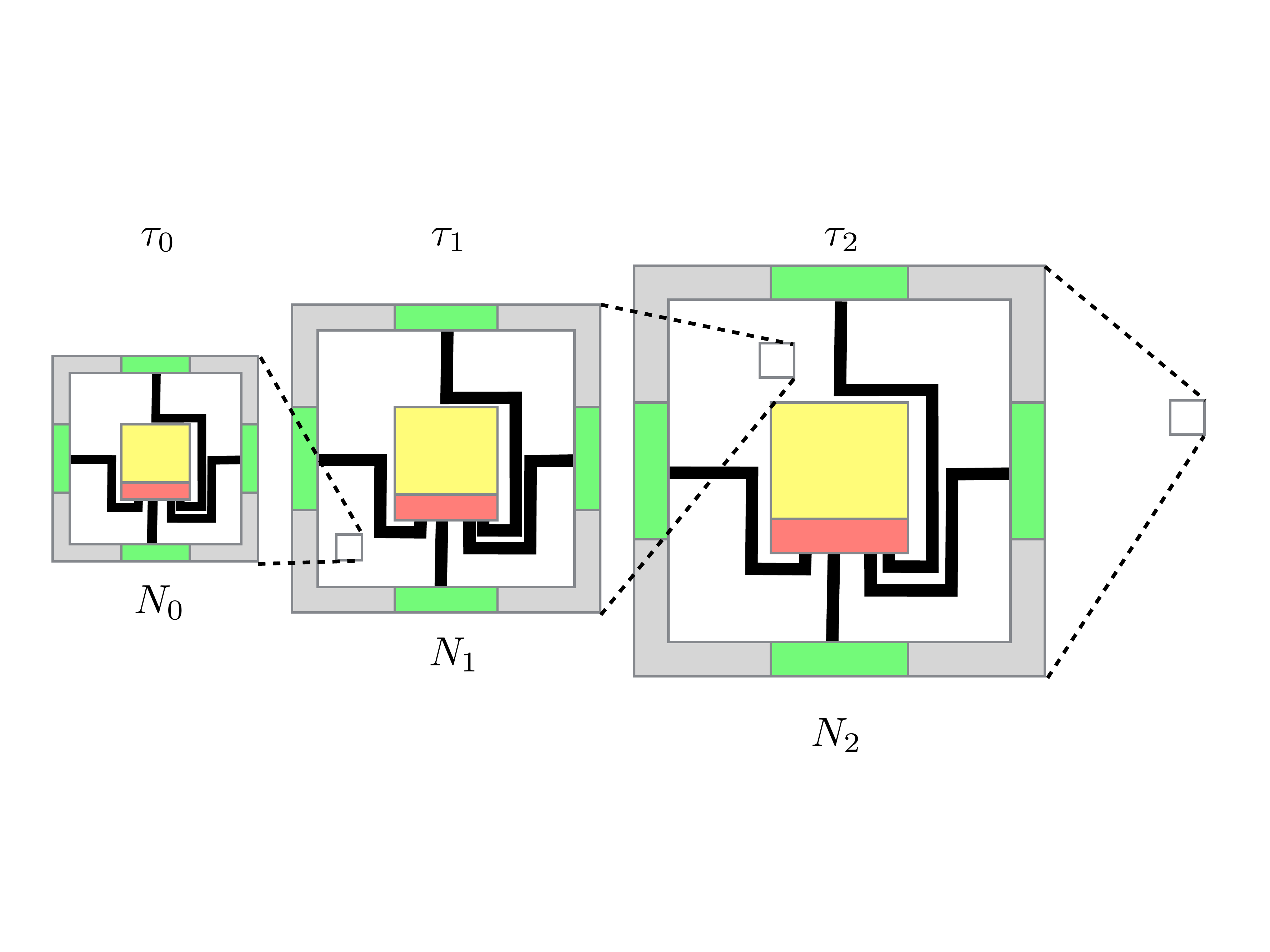}
\caption{Sequence of macro tiles which simulate the previous level tile sets.} \label{sequence of macro tiles}
\end{center}
\end{figure} 

\subsection{Simulation to check forbidden strings} \label{Simulation to check forbidden strings}

According to the previous discussion we can construct  a sequence $\{\tau_k\}$ of tile sets which simulate the tiles in the next level.
Moreover, macro tiles in any level simulate the same program which includes $\EuScript{M}_X$.
We superimpose tiles which carry alphabet $\mathcal{A}$ on tiles in $\tau_0$.
The new layer is called the input layer and the other one is called the computation layer. 
Then $\EuScript{M}_X$ in the macro tiles of $\tau_0$ can access the input layer and it checks whether the input is forbidden or not. 
However it is not clear how the programs in the macro tiles of higher level tile sets know the input.
We distribute the infinite strings in the input layer in the following way.
We call by a {\it level $k$ macro tile} a macro tile consists of tiles in $\tau_k$.
Consider a $\tau_0$-tiling and zoom out to see $\tau_k$-tiling.
Let $\ell_k=N_0 N_1\cdots  N_{k-1}$.
Then a level $k$ macro tile is represented by $N_k$ tiles consisting of $\ell_k\times \ell_k$ tiles in $\tau_0$.
Each $\ell_k\times \ell_k$ tile represents a level $k-1$ macro tile.
We distribute $\ell_k$ bits in the input string to $N_k$ level $k-1$ macro tiles in the following way:
The $i$th bit from the left is distributed to the $i$th level $k-1$ macro tile from the bottom.
See Fig. \ref{distribution of an input}.
Not only one bit but also string can be distributed.
Moreover we can change the length of the distributed strings, while it should be very short with regard to the size of the macro tile to which it is 
distributed.
Hence the main program in each macro tile simulates the distributed substring of the input.
Note that the way of distributing does not depend on inputs.
By the above discussion, we can obtain the set of tilings whose input layers are the vertical extension of $X$.

\begin{figure}[htbp]
\begin{center}
\includegraphics[width=300pt]{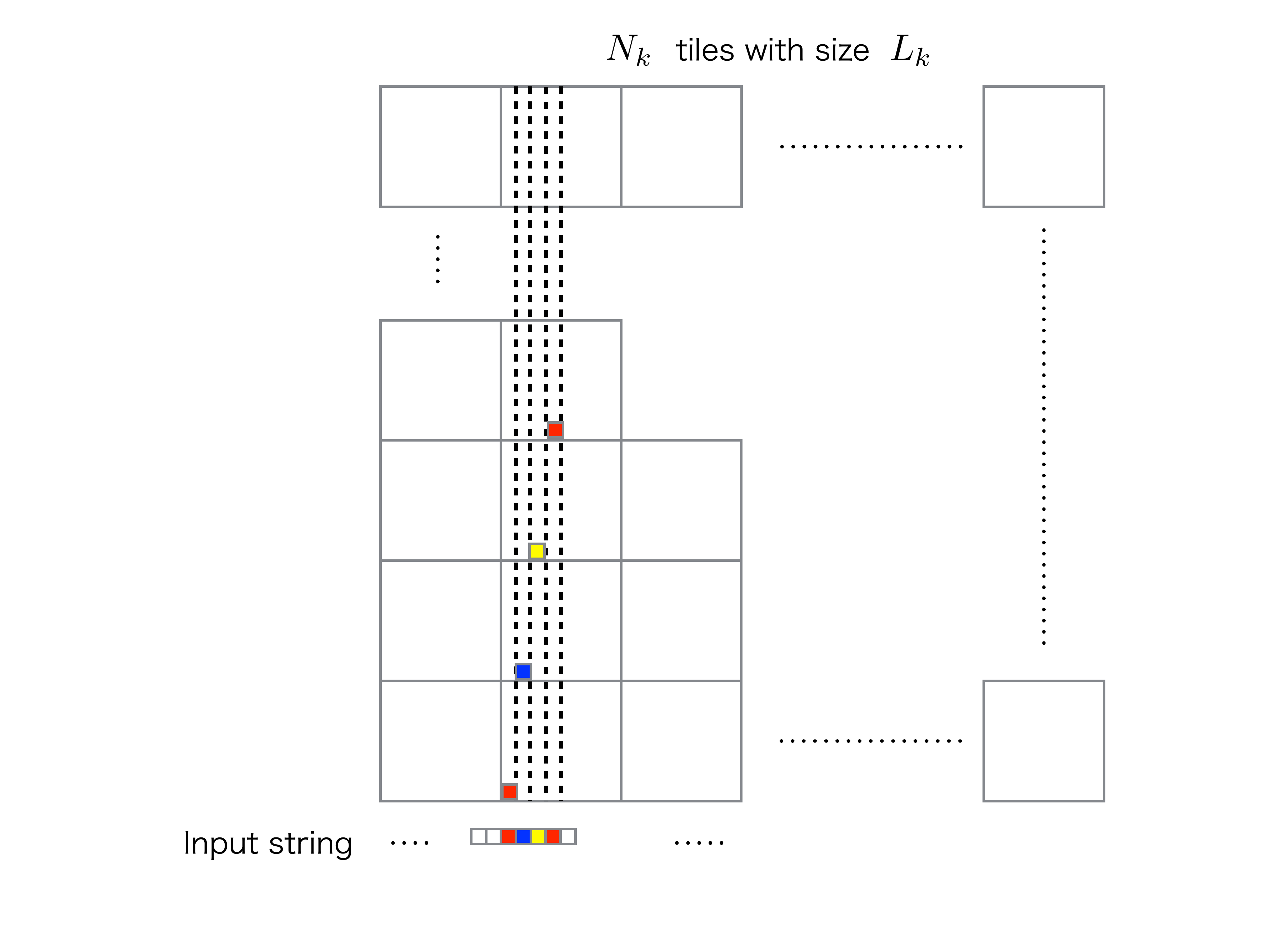}
\caption{Sequence of macro tiles which simulate the previous level tile sets.} \label{distribution of an input}
\end{center}
\end{figure} 

We now prove Proposition \ref{computation layer}.
\begin{proof}[Proof of Proposition \ref{computation layer}]
Fix $n\geq 1$ and $\omega\in \mathcal{P}_{X,n}$.
Take $p\in \pi^{-1}\, \widehat{\omega}\cap \mathcal{P}_{Y,n}$ and $\omega' \in \mathcal{P}_{X,n}$.
Consider a macro tile in the computation layer of $p$ whose size is the maximum size included in the $n\times n$ square.
The macro tile simulates a very short substring of $\omega$. 
Replace the input layer of $p$ by $\widehat{\omega}'$.
Then the macro tile simulate the substring of $\omega'$ in the same position and  with the same length as the simulated substring of $\omega$.
Since $\omega'$ is an admissible string of $X$, the simulation does not halt.
Hence we have a new macro tile.
This change of the macro tile influences macro tiles in higher levels to change the colors on the boundary and their simulations are also renewed.
Since there is a one-to-one correspondence between inputs and the patterns of simulation,
a map from the patterns in the computation layer of $\pi^{-1}\, \widehat{\omega}\,\cap\, \mathcal{P}_{Y,n}$ to $\pi^{-1}\, \widehat{\omega}'\cap \mathcal{P}_{Y,n}$ is 
injective.
Since $\omega$ is arbitrary, this proves that the number of patterns in $\pi^{-1}\, \widehat{\omega}\,\cap\mathcal{P}_{Y,n}$ does not depend on the input $\omega$.
The number of possible patterns in the computation layer depends on the positions of the macro tile with the maximum size included in a square of size 
$(2n-1)\times (2n-1)$. Hence $c_n$ is at most $(2n-1)^2$.

\begin{figure}[htbp]
\begin{center}
\includegraphics[width=300pt]{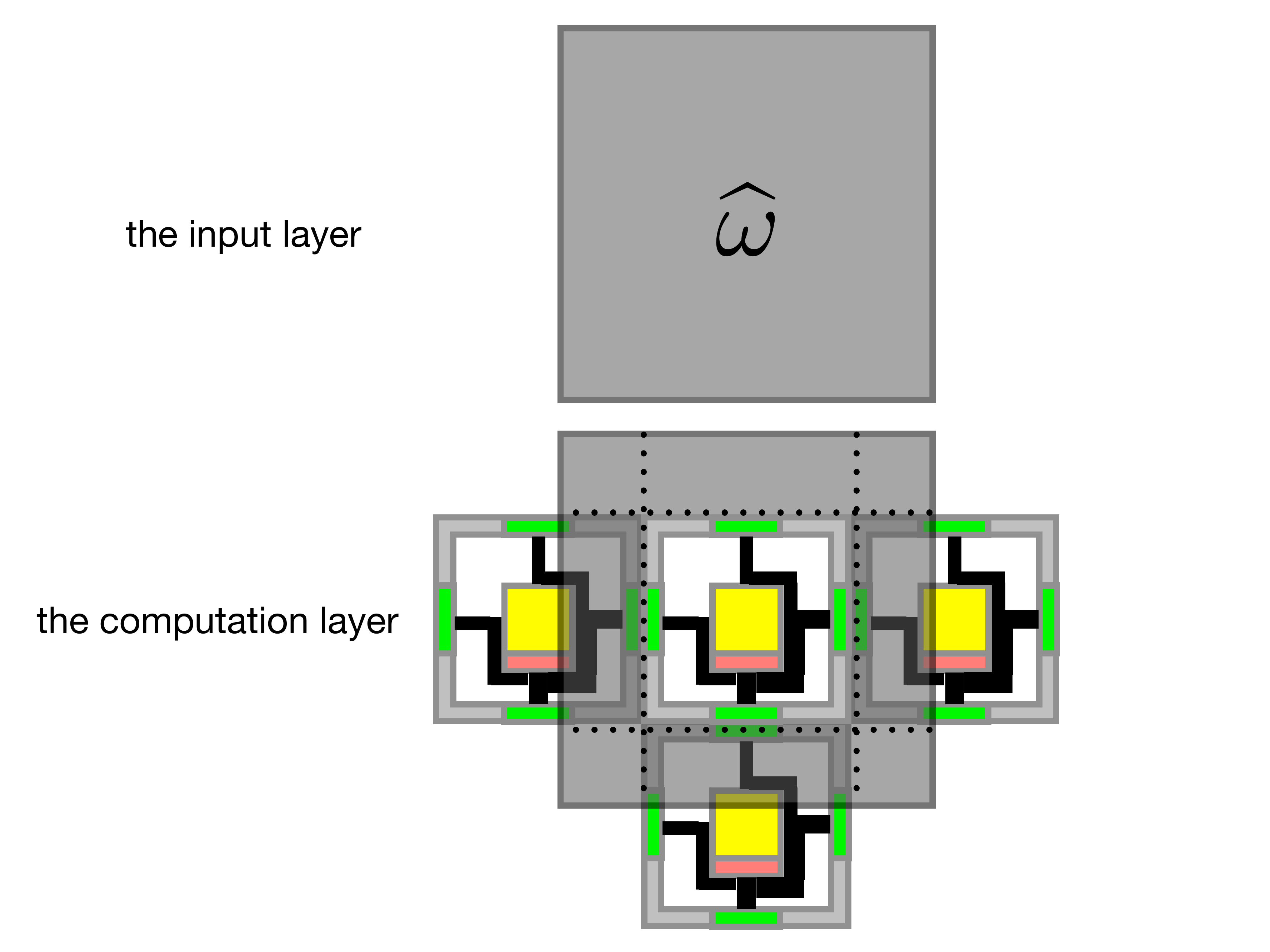}
\caption{An element $p\in \pi^{-1}\widehat{\omega}\cap \mathcal{P}_{n, Y}$.
} \label{globally admissible pattern}
\end{center}
\end{figure} 
\end{proof}

\section{Construction of a some 1D effective subshift} \label{Construction-of-a-base-system}

We construct a particular 1D subshift which we will imbed into a two-dimensional SFT by using Theorem \ref{imbedding}.
Then we define a finite-range interaction which somewhat penalizes the admissible patterns of this SFT. 

We consider the alphabet $\Sigma=\{-1, 0, +1\}$.
Let $\Sigma_+=\{0, +1\}$ and $\Sigma_-=\{0, -1\}$. 
We construct a sequence $\{X_{k}\}_{k\geq1 }$ of SFTs by giving the sequence $\{F_k\}$ of forbidden sets. 
 
We will consider an increasing sequence $\{\ell_k\}$ of lengths of forbidden strings
and a decreasing sequence $\{r_k\}$ of frequencies of $0$ in forbidden strings.
We will choose these sequences to satisfy $\lim_{k\to\infty}\ell_k=\infty$ and $\lim_{k\to\infty} r_k=0$.
We will give the precise conditions on $\{\ell_k\}$ and $\{r_k\}$ in Section \ref{Estimates-on-admissible-patterns}.
Here we explain the basic idea to make a sequence of SFTs by controlling the frequency of $0$ in forbidden sets.
For each $n$ set $\Sigma^n=\Sigma^{\{0,1,\ldots, n-1\}}$.
For $\omega\in \Sigma^n$ denote by $f_0$ the frequency with which $0$ appears in $\omega$, i.e., 
\[
f_0(\omega)=\frac{1}{n}\big|\big\{i\in \{0,1, \ldots, n-1\}: \omega_i=0\big\}\big|.
\]
Let 
\[
F_{1}=\Sigma^{2\ell_1-1}\,\setminus\, \Sigma_+^{2\ell_1-1}
\cup \big\{\omega\in \Sigma_+^{2\ell_1-1} : f_0(\omega) \geq r_1\big\}.
\]
The first set forbids the strings including $-1$, while the second set forbids the strings with many zeros.
Similarly define $F_2$ by
\[
F_{2}=\Sigma^{2\ell_2-1}\,\setminus\, \Sigma_-^{2\ell_2-1}
\cup\big\{\omega\in \Sigma_-^{2\ell_2-1} : f_0(\omega) \geq r_2\big\}.
\]
For $m\geq2$ define $F_{2m-1}$ by 
\begin{align*}
& F_{2m-1}
=\Sigma^{2\ell_{2m-1}-1}\,\setminus\, \Sigma_+^{2\ell_{2m-1}-1}
\cup \big\{\omega\in \Sigma_+^{2\ell_{2m-1}-1}: f_0(\omega) \geq r_{2m-1}\big\}\\
&\cup \left\{\omega\in \Sigma_+^{2\ell_{2m-1}-1}:  \text{there exists a subsequence}\ \eta\ \text{of}\ \omega\ \text{s.t.}\ \eta\in \bigcup_{i=1}^{m-1} F_{2i-1}\right\}
\end{align*}
and 
$F_{2m}$ by 
\begin{align*}
& F_{2m}=\Sigma^{2\ell_{2m}-1}\,\setminus\, \Sigma_-^{2\ell_{2m}-1}
\cup \big\{\omega\in \Sigma_-^{2\ell_{2m}-1}: f_0(\omega) \geq r_{2m}\big\}\\
&\cup \left\{\omega\in \Sigma_-^{2\ell_{2m}-1}: \ 
\text{there exists a subsequence}\ \eta\ \text{of}\ \omega\ \text{s.t.}\ \eta\in \bigcup_{i=1}^{m-1} F_{2i}\right\}.
\end{align*}
Let $X_k=X_{F_k}$. By construction we have 
\begin{align*}
& X_1\supset X_3\supset \cdots \supset X_{2m-1}\supset \cdots\\
& X_2\supset X_4\supset \cdots \supset X_{2m}\supset \cdots.
\end{align*}
Let $X_+=\bigcap_{m\geq1} X_{2m-1}$, $X_-=\bigcap_{m\geq1} X_{ 2m}$ and $X=X_+\cup X_-$.
Note that $X_\pm\neq\emptyset$ since $\pm1^\infty \in X_\pm$.
By construction we have
$\mathcal{P}_{X_+, \ell_{2m-1}}=\Sigma_+^{2l_{2m-1}-1}\,\setminus\, F_{2m-1}$ and 
$\mathcal{P}_{X_-, \ell_{2m}}=\Sigma_+^{2l_{2m}-1}\,\setminus\, F_{2m}$
for $m\geq1$.
We will choose $\{\ell_k\}$ and $\{r_k\}$ in such a way that $h_{{\scriptscriptstyle top}}(X_1)>h_{{\scriptscriptstyle top}}(X_2)>h_{{\scriptscriptstyle top}}(X_3)> \cdots >h_{{\scriptscriptstyle top}}(X_{2m-1})>h_{{\scriptscriptstyle top}}(X_{2m})>h_{{\scriptscriptstyle top}}(X_{2m+1})>\cdots$ and  to keep this kind of condition after imbedding in two dimensions.

\begin{figure}[htbp]
\begin{center}
\includegraphics[width=300pt]{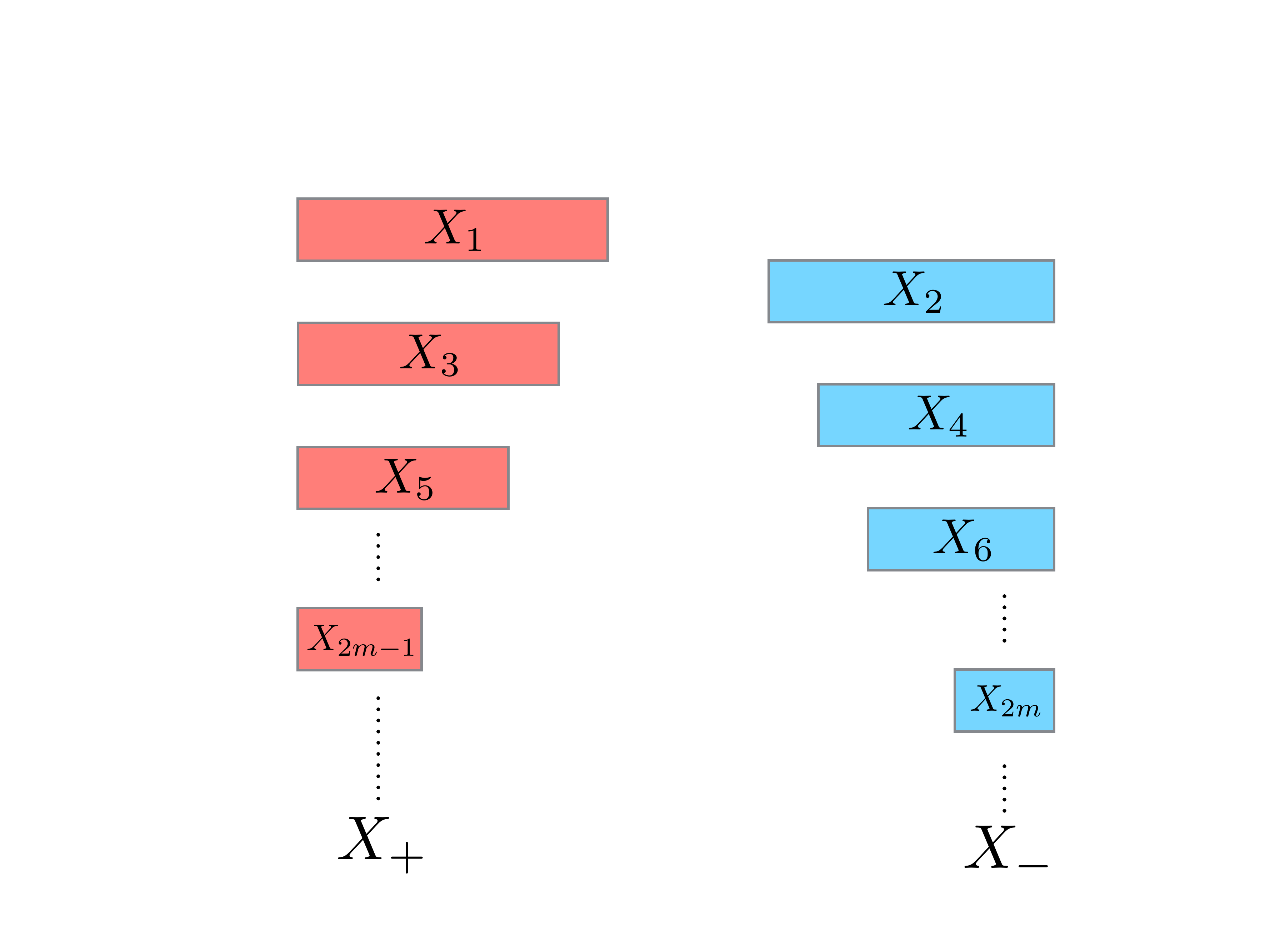}
\caption{Sketch of the sequences of nested subshifts of finite type.} \label{sequence of subshifts.}
\end{center}
\end{figure} 

Applying Theorem \ref{imbedding} to the subshift $X$,
we know that there exist an alphabet $\breve{B}\subset \Sigma\times  \mathcal{C}$
and a two-dimensional SFT $Y$ over $\breve{B}$.

We need to blow up the number of patterns in $Y$.
We replace $\Sigma$ in $\breve{B}$ by
$\widetilde{\Sigma}=\{\widetilde{\omega}=(\omega, s) \in \Sigma\times\{0, \tilde{0}, \ast\} : s=\ast\ \text{if}\ \omega\neq 0\}$.
Let $\widetilde{\mathcal{B}}$ be the corresponding alphabet.
We have two possibilities when $\omega=0$, which blows up the number of patterns corresponding to the numbers of $0$.
Denote by $\widetilde{Y}$ the corresponding SFT.
Let $\widetilde{Y}_\pm=\widetilde{\pi}^{-1} X_\pm \cap \widetilde{Y}$ where $\widetilde{\pi}$ is the projection from $\widetilde{\mathcal{B}}$ to $
\Sigma$.

Let $F$ be the (finite) set of forbidden patterns for $\widetilde{Y}$. Elements in $F$ are ``cross-shaped''. 
Define $\Lambda_F=\{(0,0), (\pm1,0), (0, \pm 1)\}$, and a finite-range interaction $\Phi$ by
\[
\Phi_\Lambda(x)=
\begin{cases}
-|\Lambda| & \text{if} \quad \Lambda=\Lambda_F\;\;\text{and}\;\; x_{\Lambda_F}\in F\\
0 & \text{otherwise}. 
\end{cases}
\]
Then, the corresponding locally constant potential $\phi: \widetilde{\mathcal{B}}^{\,\Z^2}\rightarrow \R$ is (recall \eqref{def-phi-from-Phi})
\[
\phi(x)=
\begin{cases}
-1 & \text{if} \quad x_{\Lambda_F}\in F\\
0 & \text{otherwise}. 
\end{cases}
\]
That $\phi$ is locally constant means that for every $p\in F$, $\phi(x)=\phi(x')$ whenever $x,x'\in [p]$. 
Observe that $\phi\equiv0$ on $\widetilde{Y}$, since configurations in $\widetilde{Y}$ have no pattern from $F$.

\section{Estimates on admissible patterns} \label{Estimates-on-admissible-patterns}

We start with the conditions we have to impose on $\{\ell_k\}$ and $\{r_k\}$.
For each $k\geq1$ we require that
\[
\tag{S1}
(2\ell_k-1)\, r_{k} \geq1,
\]
which ensure that there exists at least one admissible string with $0$ of any size.
For $t\in \left[0,1\right]$, let
\begin{equation}\label{def-binary-entropy}
H(t)=-t\log t -(1-t)\log (1-t)
\end{equation}
with the usual convention $0\log0=0$.
Let $\ell_1=2^2$ and $r_1=2^{-1}$. 
Define $\ell_{k+1}$ and $r_{k+1}$ inductively by the following conditions:
\[
\tag{S2}
2^{-1}r_k \log 2 \geq 10 H(r_{k+1})
\]
\[
\tag{S3}
(4\ell_k-2)^{-1}\geq 10\left(r_{k+1} + 2(2\ell_{k+1}-1)^{-2} \log (2\ell_{k+1}-1)\right)
\]
\[
\tag{S4}
\ell_{k+1}\geq 2^{4\ell_k}.
\]
Note that (S3) implies
\[
(4\ell_k-2)^{-1}\geq 10\left(r_{k+1} + (2\ell_{k+1}-1)^{-2} \log c_{\ell_{k+1}}\right)
\]
since $c_n\leq (2n-1)^2$.
Since we can rewrite conditions (S1), (S2), (S3) and (S4) by using recursive functions, 
the subshift $X$ we defined in previous section is effective. 

The input layer of an admissible pattern of $\widetilde{Y}$ can contain forbidden strings
since macro tiles on the computation layer only checks very short substrings of the input.
Such input patterns are said to be locally admissible.
For later use we define the globally admissible set.
For $n\geq1$ define the globally admissible set $G_{\widetilde{Y}, n}$ of $\widetilde{Y}$ with size $n$ by
\[
G_{\widetilde{Y}, n}=\pi^{-1}\, \mathcal{P}_{\widehat{X}, n} \cap \mathcal{P}_{\widetilde{Y},\,n}.
\]
Since we restrict the admissible set of $\widetilde{Y}$ to the elements whose input layer consists of admissible strings of $X$,
an element of $G_{\widetilde{Y}, n}$ can be extended to bigger and bigger squares.
 
Define $m_k=(2\ell_k-1) (2\ell_{k-1}-1)^{-1}$. 
Then an admissible string of $X_-$ with length $2\ell_k-1$ consists of $m_{k}$ admissible strings of $X_-$ with length $2\ell_{k-1}-1$.
However $\mathcal{P}_{X_-,\, \ell_{k-1}}^{m_k} \neq \mathcal{P}_{X_-, \,\ell_k}$, since we may find forbidden strings in the concatenated part.
Let $k$ be even.
Let $B_{+, k-1}$ be the set of admissible strings with more than one $0$, i.e.,
$B_{+, \, k-1} = \mathcal{P}_{X_+,\, \ell_{k-1}} \setminus  \{ +1^{\ell_{k-1}}\}$.
Let $C_{+,\, k}$ be the set of strings obtained by concatenating alternatively strings from $B_{+, \,k-1}$ and the string $+1^{2\ell_{k-1}-1}$:
\begin{align*}
C_{+,\, k} = & \left\{\omega_1 \cdots \omega_{m_k}\in \mathcal{P}_{X_+,\, \ell_{k-1}}^{m_k}: 
\omega_m\in B_{+,k-1}\ \text{if}\ m\ \text{is odd}\right. \\
& \left. \;\;\qquad\qquad\qquad\qquad\qquad\text{and}\ \omega_m=(+1)^{2\ell_{k-1}-1}\ \text{if}\ m\ \text{is even}\right\}\,.
\end{align*}
Since the string $(+1)^{2\ell_{k-1}-1}$ appears between strings from $B_{+,\,k-1}$, no forbidden string appears in the concatenated parts.
Hence we have $C_{+,\, k}\subset \mathcal{P}_{X_+, \, \ell_k}$.
We define $B_{-,\, k}$ and $C_{-,\, k}$ for odd $k$ in the same way.

We use the following lemma to show Proposition \ref{pattern lemma}.
\begin{lemma}
We have 
\begin{equation}
\label{string_relation_odd}
|\mathcal{P}_{X_-,\, \ell_k}|^{10}\leq |C_{+,\, k}| \quad\text{if}\;\;k\;\;\text{is odd}
\end{equation}
and 
\begin{equation}
\label{string_relation_even}
|\mathcal{P}_{X_+,\, \ell_k}|^{10}\leq |C_{-, \, k}| \quad\text{if}\;\;k\;\;\text{is even}\,.
\end{equation}	
\end{lemma}
\begin{proof}
Assume $k$ is odd. By (S2) we have the following.
\begin{align*}
|\mathcal{P}_{X_-,\ell_k}|^{10}
& \leq \left(\sum_{i=0}^{(2\ell_k-1)\, r_k} \binom{2\ell_k-1}{i} \right)^{10} 
\leq \left(\e^{(2\ell_k-1)H(r_k)}\right)^{10}\\
&= \e^{10(2\ell_k-1) H(r_{k})}\\
& \leq \e^{2^{-1} (2\ell_k-1)\, r_{k-1}}.
\end{align*}
Now
\begin{align*}
|C_{+, k}| 
&=|B_{+,\, k-1}|^{2^{-1}m_k}
=\left(\sum_{i \leq (2\ell_{k-1}-1)\,r_{k-1}}\binom{ 2\ell_{k-1}-1}{i}-1\right)^{2^{-1}m_k}\\
& \geq \left(\sum_{i \leq (2\ell_{k-1}-1)\, r_{k-1}}
\binom{ (2\ell_{k-1}-1)\,r_{k-1}}{i}\right)^{2^{-1}m_k}\\
& =\e^{ (2\ell_{k-1}-1)\, r_{k-1} 2^{-1}m_k}
= \e^{2^{-1}(2\ell_k-1)\, r_{k-1}}\,.
\end{align*}
Hence we get \eqref{string_relation_odd}.
The proof of \eqref{string_relation_even} is very similar and thus left to the reader.
\end{proof}

For each $n\geq1$ set the globally admissible set $G_{\widetilde{Y}_\pm, n}$ of $\widetilde{Y}_\pm$ with size $n$ by
\[
G_{\widetilde{Y}_\pm, n}=\pi^{-1}\, \mathcal{P}_{\widehat{X}_\pm, n} \cap \mathcal{P}_{\widetilde{Y},\,n}.
\]

\begin{proposition}\label{pattern lemma}
We have 
\[
|G_{\widetilde{Y}_-,\ \ell_k}|^{10}\leq |G_{\widetilde{Y}_+,\ \ell_k}|\quad \text{if}\;\;k\;\;\text{is odd}
\]
and
\[
|G_{\widetilde{Y}_+,\ \ell_k}|^{10}\leq |G_{\widetilde{Y}_-,\ \ell_k}|\quad \text{if}\;\;k\;\;\text{is even}.
\]
\end{proposition}

\begin{proof}
By Proposition \ref{computation layer} we have 
\begin{align*}
|G_{\widetilde{Y}_{\pm}, \ell_k}|
&=\sum_{p\,\in Y_{\pm,\, \ell_k}} \e^{\text{number of}\;0\text{'s in}\;\pi(p)}\\
& =c_{\ell_k} \sum_{\red{\widehat{\omega}\,\in\, \scaleto{\mathcal{P}_{\widehat{X}_\pm,\ell_k}}{7pt}}}
\e^{\text{number of}\;0\text{'s in}\;\widehat{\omega}}\\
& = c_{\ell_k} \sum_{\omega\,\in\, \mathcal{P}\!_{X_\pm, \,\ell_k}} \e^{f_0(\omega)(2\ell_k-1)^2}.
\end{align*}
Assume $k$ is odd.
By \eqref{string_relation_even} we have
\begin{align*}
|G_{\widetilde{Y}_{-}, \ell_k}|^{10}
&=\left(c_{\ell_k}\sum_{\omega\,\in\, \mathcal{P}_{X_+\!, \ell_k}} \e^{f_0(\omega)(2\ell_k-1)^2}\right)^{10}\\
& \leq | \mathcal{P}_{X_-, \ell_k} |^{10} \e^{10\, r_k(2\ell_k-1)^2+10\log c_{\ell_k}}\\
& \leq |C_{+\!,\, k}| \, \e^{10(r_k(2\ell_k-1)+(2\ell_k-1)^{-1}\log c_{\ell_k}) (2\ell_k-1)}.
\end{align*}
Note that every $\omega$ in $C_{+\!, k}$ satisfies $f_0(\omega) (2\ell_k-1) \geq 2^{-1} m_k$.
By (S3) we have
\begin{align*}
f_0(\omega)(2\ell_k-1)^2
& \geq \frac{2\ell_k-1}{2(2\ell_{k-1}-1)}(2\ell_k-1)\\
&\geq 10  (r_k(2\ell_k-1) + (2\ell_k-1)^{-1}\log c_{\ell_k}) (2\ell_k-1)
\end{align*}
for $\omega \in C_{+\!,k}$.
Hence we have
\begin{align}
|G_{\widetilde{Y}_{-}, \ell_k}|^{10}
\leq \sum_{\omega\, \in\, C_{+\!,k}}\e^{f_0(\omega)(2\ell_k-1)^2}
\leq c_{\ell_k} \sum_{\omega\,\in\, \mathcal{P}_{X_+\!, \ell_k}} \e^{f_0(\omega)(2\ell_k-1)^2}
=|G_{\widetilde{Y}_+\!, \ell_k}|.
\label{GYplus}
\end{align}	
Therefore, for odd $k$, the proof is finished. We can do very similar calculations for even $k$.
\end{proof}

\section{Proof of the main theorem } \label{Non-convergence}

We now prove Theorem \ref{maintheorem}. 
For each $\beta>0$, we pick an arbitrary equilibrium state for $\beta\phi$ (there can be several) and we consider the resulting one-parameter
family $(\mu_{\beta})_{\beta>0}$. 
Let $M_{k}$ be the minimum size of the macro-tiles whose simulation checks all substrings with length less than $\ell_k$.
For each $k\geq 1$ we denote by $\theta(k)$ the smallest index which satisfies $2\ell_{\theta(k)}-1\geq 2 M_k$.
By this choice at least one macro tile with size $M_k$ should appear in the box with size $\ell_{\theta(k)}$.

Our main theorem follows from the following proposition and the rest of this paper is dedicated to prove it.


\begin{proposition} \label{oscillation}
Take an arbitrary $\delta\in (0,1]$. Then for all $k$ large enough we have
\begin{align*}
\mu_{\beta_{k}}\left( [G_{\widetilde{Y}_-,\, \ell_{k+1}}] \right) & \geq 1-\delta \quad \text{if}\quad k\;\; \text{is even},\\
\mu_{\beta_{k}}\left([G_{\widetilde{Y}_+,\, \ell_{k+1}}] \right)& \geq 1-\delta \quad \text{if}\quad  k\;\; \text{is odd}.
\end{align*}
\end{proposition}
Invoking the ergodic decomposition \cite[Chapter 14]{Georgii}, 
we can restrict ourselves to ergodic equilibrium states in the proof of Proposition \ref{oscillation}.

Since $\bigcap_{k\geq1} [G_{\widetilde{Y}_-, \ell_k}]=\widetilde{Y}_-$, $\bigcap_{k\geq1}[G_{\widetilde{Y}_+, \ell_k}]=\widetilde{Y}_+$
and $\widetilde{Y}_-\cap \widetilde{Y}_+=\emptyset$, the proposition shows that the one-parameter family $(\mu_\beta)_{\beta>0}$ does not converge. 

To prepare the proof of the above proposition, we need the following three lemmas.

\begin{lemma}\label{weight lemma on the expectation}
Let $C=\log\big|\widetilde{\mathcal{B}}\big|$.
Then $\int \phi \dd\mu_{\beta}\geq -C\beta^{-1}$ for any $\beta>0$.
\end{lemma}
\begin{proof}
There exists at least one shift-invariant measure $\mu$ whose support is contained in $\widetilde{Y}$.
Since $\phi\equiv 0$ on $\widetilde{Y}$, this implies that $\int \phi \dd\mu=0$. Hence, since $h(\mu)\geq0$, we get
\[
P(\beta \phi)\geq h(\mu)+\beta \int \phi\dd\mu \geq 0.
\]

Since $\mu_{\beta}$ is an equilibrium state for $\beta\phi$, we have
\[
h(\mu_{\beta})+\beta\int \phi \dd\mu_{\beta}=P(\beta\phi)\geq0.
\]
Therefore
\[
\int \phi \dd\mu_{\beta}\geq-\beta^{-1}h(\mu_{\beta})
\geq  -\beta^{-1}\log \big|\widetilde{\mathcal{B}}\big|. 
\]
\end{proof}

Given $n\geq1$ and a function $\psi: \widetilde{\mathcal{B}}^{\,\Z^2}\rightarrow \R$, let $S_n\psi=\sum_{i\in \Lambda_n} \psi \circ \sigma^i$.

Let $C'=2C$ and for $n,k \geq1$ define $E_{n,k}$ as the set of configurations $y\in \widetilde{\mathcal{B}}^{\,\Z^2}$ satisfying the following conditions:
\begin{align}
\frac{S_{n+\ell_{\theta(k+1)}-1} \phi(y)}{\lambda_{n+\ell_{\theta(k+1)}-1}}  \geq \int \phi \dd\mu_{\beta_{k}} -C \beta_{\ell_{k}}^{-1}
\label{Sn1}
\end{align}
and 
\begin{align}
\frac{1}{\lambda_n} S_n\chi_{[G_{\widetilde{Y}_-, \ell_{k+1}}]} (y) \leq \mu_{\beta_{k}} \left([G_{\widetilde{Y}_-,\ell_{k+1}}]\right) + C'2^{-2\ell_{k}}.
\label{Sn2}
\end{align}

\begin{lemma} \label{estimate of Enk}
Take $k\geq1$ and $\varepsilon>0$. 
Let $\mu_{\beta_{k}}$ be an ergodic equilibrium state for $\beta_{k}\phi$.
Then for all $n$ large enough
\begin{align}
\mu_{\beta_{k}}(E_{n,k})>1-\varepsilon. \label{probability of Enk}
\end{align}
\end{lemma}
\begin{proof}
Since $\mu_{\beta_{k}}$ is ergodic, we have
\begin{align*}
\lim_{n\to\infty} \frac{1}{\lambda_n} S_n\phi(y)
&= \int \phi \dd\mu_{\beta_{k}}\\
\lim_{n\to\infty}\frac{1}{\lambda_n}S_n\chi_{[G_{\widetilde{Y}_-, \ell_{k+1}}]}(y)
&=\mu_{\beta_{k}} \left([G_{\widetilde{Y}_-,\ell_{k+1}}]\right)
\end{align*}
for $\mu_{\beta_{k}}$ almost every point $y$,
which completes the proof.
\end{proof}

We estimate the number of globally admissible patterns in the configurations of $E_{n,k}$.
\begin{lemma} \label{Enk}
For every $k$ there exists $N\geq1$ such that for every $n\geq N$
\begin{align}
\frac{1}{\lambda_n} \#\Big\{i\in \Lambda_n: y_{i+\Lambda_{\ell_{k+1}}\in\, G_{\widetilde{Y},  \ell_{k+1}}}\Big\}
=\frac{1}{\lambda_n} S_n \chi_{[G_{\widetilde{Y}, \ell_{k+1}}]}(y)
>1-C'2^{-2\ell_{k}}
\label{patterns-of-Enk}
\end{align}
for every $y\in E_{n,k}$.
\end{lemma}

\begin{proof}
Take $N\geq1$ such that $\lambda_n^{-1} \lambda_{n+\ell_{\theta(k+1)}-1} \lambda_{\ell_{\theta(k+1)}}<2^{\ell_{k}}$,
for all $n\geq N$.
Take $n\geq N$ and $y\in E_{n,k}$.
Since by definition $[G_{\widetilde{Y}, \ell_{k+1}}]\supset [\mathcal{P}_{\widetilde{Y}, \ell_{\theta(k+1)}}]$, we have
\[
1-\frac{1}{\lambda_n} S_n\chi_{[G_{\widetilde{Y}, \ell_{k+1}}]} (y)
=\frac{1}{\lambda_n} S_n \chi_{\widetilde{\mathcal{B}}^{\Z^2}\setminus[G_{\widetilde{Y}, \ell_{k+1}}]} (y)\\
\leq \frac{1}{\lambda_n} S_n \chi_{\widetilde{\mathcal{B}}^{\Z^2}\setminus[\mathcal{P}_{\widetilde{Y}, \ell_{\theta(k+1)}}]} (y).
\]
$S_n \chi_{\widetilde{\mathcal{B}}^{\Z^2}\setminus[\mathcal{P}_{\widetilde{Y}, \ell_{\theta(k+1)}}]}  (y)$ is the number of positions in $\Lambda_n$
for which a non-admissible pattern with size $\ell_{\theta(k+1)}$ appears.
Let $i\in \Lambda_{n+\ell_{\theta(k+1)}-1}$ be a position where a forbidden pattern appears, that is, $y_{\Lambda_F+i}\in F$.
Then the patterns in boxes with size $\ell_{\theta(k+1)}$ including $i$ are non-admissible.
The number of such boxes is bounded by $\lambda_{\ell_{\theta(k+1)}}$. 
Since the number of positions in $\Lambda_{n+\ell_{\theta(k+1)}-1}$ where a forbidden pattern appears is $-S_{n+\ell_{\theta(k+1)}-1} \phi(y)$, 
we have
\begin{align*}
& \frac{1}{\lambda_n} S_n \chi_{\widetilde{\mathcal{B}}^{\Z^2}\setminus[\mathcal{P}_{\widetilde{Y}, \ell_{\theta(k+1)}}]} (y)\\
&\leq- \frac{\lambda_{n+\ell_{\theta(k+1)}-1}}{\lambda_n}\frac{1}{\lambda_{n+\ell_{\theta(k+1)}-1}} S_{n+\ell_{\theta(k+1)}-1} \phi (y) \times \lambda_{\ell_{\theta(k+1)}}.
\end{align*}
See Fig. \ref{non admissible patterns}.
By the choice of $n$, \eqref{Sn1} and Lemma \ref{weight lemma on the expectation}, we have
\begin{align*}
1-\frac{1}{\lambda_n} S_n\chi_{[G_{\widetilde{Y}, \ell_{k+1}}]} (y)
&\leq 2^{\ell_{k}}\left(-\int \phi \dd\mu_{\beta_{\ell_{k}}} +C\beta_{\ell_{k}}^{-1}\right)\\
&\leq 2C\beta_{\ell_{k}}^{-1}2\, ^{\ell_{k}}=C' \, 2^{-2\ell_{k}}.
\end{align*}
\end{proof}

\begin{figure}[htbp]
\begin{center}
\includegraphics[width=200pt]{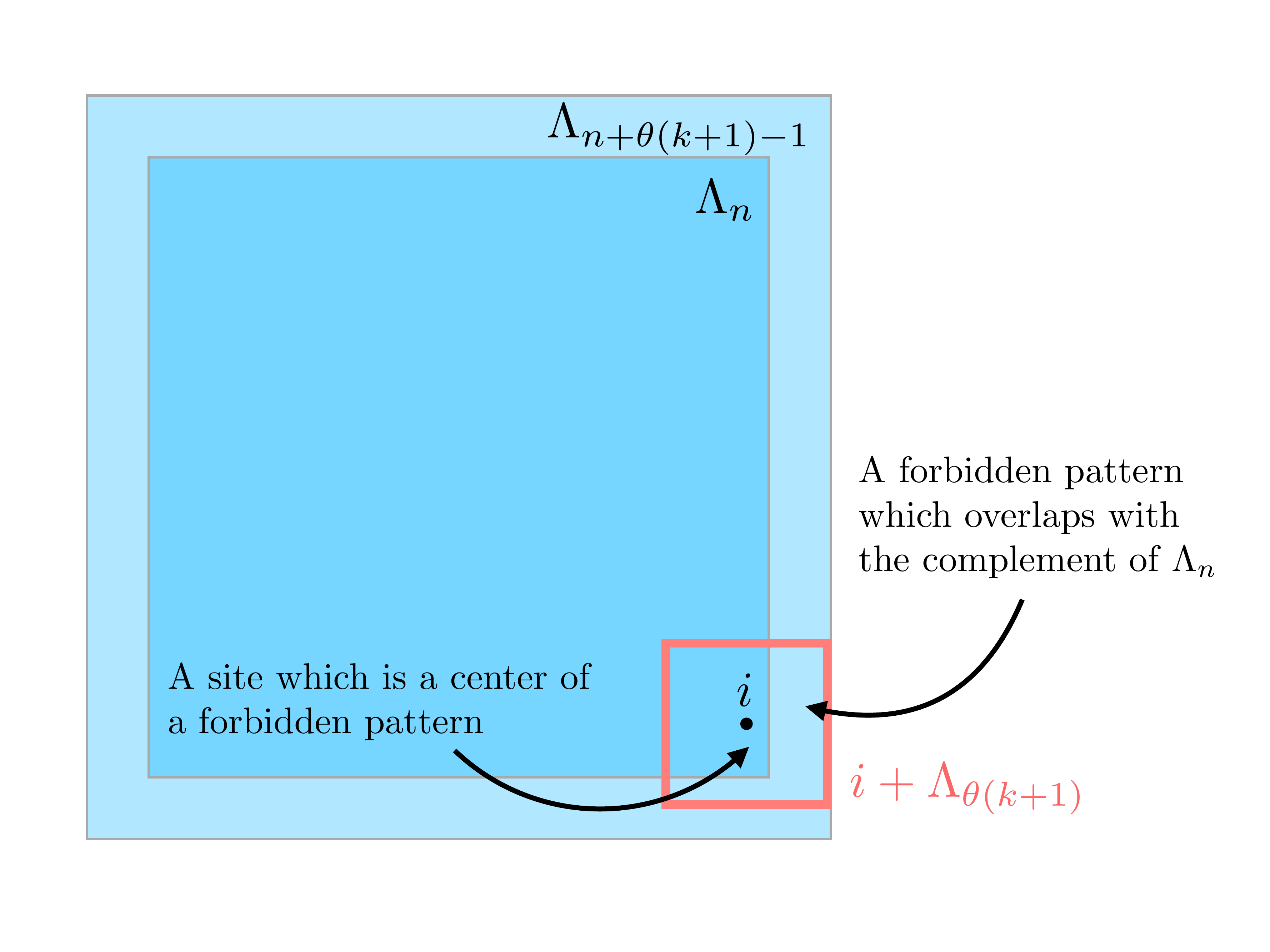}
\caption{Estimate of the number of non-admissible patterns.} \label{non admissible patterns}
\end{center}
\end{figure} 
 
\begin{proof}[Proof of Proposition \ref{oscillation}]
By contradiction. We suppose that there exist $\delta\in (0,1]$ and $k$ large enough such that
\begin{align}
\mu_{\beta_{k}}\left( [G_{\widetilde{Y}_-, \ell_{k+1}}] \right)\leq 1-\delta \quad \text{if}\  k\ \text{is even}, \label{contradiction}\\
\mu_{\beta_{k}}\left([G_{\widetilde{Y}_+, \ell_{k+1}}] \right)\leq 1-\delta \quad \text{if}\  k\ \text{is odd}. \nonumber
\end{align}
We only deal with the case when $k$ is even, since the other case is similar. 
Without loss of generality we may assume $k$ is large enough to satisfy $C' 2^{-2\ell_{k}}\leq \delta$
and $\Lambda_{\ell_{k+1}}\supset \Lambda_F$.
By \eqref{Sn2} and \eqref{contradiction} we have
\begin{align}
\frac{1}{\lambda_n}S_n \chi_{[G_{\widetilde{Y}_-, \ell_{k+1}}]}(y)
\leq 1-\delta + C'2^{-2\ell_{k}}
\leq 1-\left(1+\frac{1}{100}\right)\delta
\label{pattern from negative}
\end{align}
for every $y\in E_{n,k}$.

Let $h_{k+1}^-=\frac{1}{\lambda_{\ell_{k+1}}}\log\big|G_{\widetilde{Y}_-, \ell_{k+1}}\big|$.
Using  \eqref{GYplus} in Proposition \ref{pattern lemma}, we have
\[
\Big|G_{\widetilde{Y}_-, \ell_{k+1}}\Big|
\geq \sum_{\omega \in C_{-, k+1}}\e^{f_0(\omega)(2\ell_{k+1}-1)^2}
\geq \e^{\min_{\omega\in C_{-, k+1}} f_0(\omega)(2\ell_{k+1}-1)^2}.
\]
By the definition of $C_{-,\, k+1}$ we have
\begin{equation}\label{hkminus}
h_{k+1}^-=\frac{1}{\lambda_{\ell_{k+1}}}\log \Big|G_{\widetilde{Y}_-, \ell_{k+1}}\Big|
\geq\min_{\omega\in C_{-,k+1}} f_0(\omega)
\geq\frac{\lfloor \frac{m_{k+1}}{2}\rfloor}{\ell_{k+1}}
\geq O(\ell_{k}^{-1}).
\end{equation}

\begin{lemma} \label{estimation on exponential decay}
Take $N\geq1$ such that, for every $n\geq N$ and $y\in E_{n,k}$, \eqref{patterns-of-Enk} holds.
Then we have
\[
\frac{1}{\lambda_n}\log|E_{n,k}|\leq H\big(C'2^{-2\ell_{k}}\big)+ (1-\delta') h_{k+1}^- +o(h_{k+1}^-)
\]
where  $\delta'=\frac{9}{1000}\,\delta$.
\end{lemma}

\begin{proof}
Fix $n\geq N$. For $y \in E_{n,k}$ the number of positions in $\Lambda_n$ for which a pattern from $\widetilde{\mathcal{B}}^{\Lambda_{\ell_{k+1}}}\setminus G_{\widetilde{Y},\, \ell_{k+1}}$ appears is bounded above by $ \lambda_n C' 2^{-2\ell_{k}}$ by \eqref{patterns-of-Enk}: 	
\[
\lambda_n\left(1-\frac{1}{\lambda_n} S_n \chi_{[G_{\widetilde{Y}, \ell_{k+1}}]} (y)\right)
\leq \lambda_n\, C' 2^{-2\ell_{k}}.
\]
The number of possible places for such positions is bounded by
\[
\sum_{r< \lambda_n C'2^{-2\ell_{k}}}\binom{\lambda_n}{r}\leq \e^{H(C'2^{-2\ell_{k}})\lambda_n}\lambda_n^2
\]
(See Appendix \ref{appendix-binom}).
Since $G_{\widetilde{Y}\!,\, \ell_{k+1}}$ is the disjoint union of $G_{\widetilde{Y}_-, \ell_{k+1}}$ and $G_{\widetilde{Y}_+, \ell_{k+1}}$,
\eqref{pattern from negative} and \eqref{patterns-of-Enk} imply
\begin{align*}
S_n \chi_{[G_{Y_+', \ell_{k+1}}]} (y)
&=S_n \chi_{[G_{\widetilde{Y}, \ell_{k+1}}]} (y)-S_n \chi_{[G_{Y_-', \ell_{k+1}}]} (y)\\
&\geq \lambda_n\big(1-C' 2^{-2\ell_{k}}\big) -\lambda_n\left(1-\Big(1+\frac{1}{100}\Big)\delta\right)\\
& \geq \frac{\delta}{100} \lambda_n.
\end{align*}
Considering overlapping parts, the possible choices of patterns from $G_{\widetilde{Y}_+, \ell_{k+1}}$
is  bounded above by $\lambda_n \lambda_{\ell_{k+1}}^{-1}$
and is bounded below by $\frac{\delta}{100}\frac{\lambda_n}{\lambda_{\ell_{k+1}}}$.
	
Hence the number of ways to choose patterns from $G_{\widetilde{Y}, \ell_{k+1}}$ is bounded by
\begin{align*}
& \sum_{r=\frac{\delta}{100} \frac{\lambda_n}{\lambda_{\ell_{k+1}}}}^{\frac{\lambda_n}{\lambda_{\ell_{k+1}}}}
\binom{\lambda_n}{r}\, \Big|G_{\widetilde{Y}_-, \ell_{k+1}}\Big|^{ \frac{\lambda_n}{\lambda_{\ell_{k+1}}}-r}\, \Big|G_{\widetilde{Y}_+, \ell_{k+1}}\Big|^r \\
& \leq \sum_{r=\frac{\delta}{100} \frac{\lambda_n}{\lambda_{\ell_{k+1}}}}^{\frac{\lambda_n}{\lambda_{\ell_{k+1}}}}
\binom{\lambda_n}{r} \, \Big|G_{\widetilde{Y}_-, \ell_{k+1}}\Big|^{ \frac{\lambda_n}{\lambda_{\ell_{k+1}}}-r}\, \Big|G_{\widetilde{Y}_-, \ell_{k+1}}\Big|^{\frac{r}{10}}\\
& =\sum_{r=\frac{\delta}{100} \frac{\lambda_n}{\lambda_{\ell_{k+1}}}}^{\frac{\lambda_n}{\lambda_{\ell_{k+1}}}}
\binom{\lambda_n}{r}\, \Big|G_{\widetilde{Y}_-, \ell_{k+1}}\Big|^{ \frac{\lambda_n}{\lambda_{\ell_{k+1}}}-\frac{9}{10}r}\\
& \leq \Big|G_{\widetilde{Y}_-, \ell_{k+1}}\Big|^{ \frac{\lambda_n}{\lambda_{\ell_{k+1}}}(1-\frac{9}{10}\frac{\delta}{100})}
\sum_{r=\frac{\delta}{100} \frac{\lambda_n}{\lambda_{\ell_{k+1}}}}^{\frac{\lambda_n}{\lambda_{\ell_{k+1}}}}
\binom{\lambda_n}{r}\\
& \leq \Big|G_{\widetilde{Y}_-, \ell_{k+1}}\Big|^{ \frac{\lambda_n}{\lambda_{\ell_K}}(1-\delta')} \frac{\delta}{100} \left(\frac{\lambda_n}{\lambda_{\ell_{k+1}}} \right)^2 \e^{\frac{\lambda_n}{\lambda_{\ell_{k+1}}}H\big(\frac{\delta}{100}\big)}.
\end{align*}
Hence
\[
|E_{n, k}|\leq e^{H(C'2^{-2\ell_{k}})\lambda_n}\lambda_n^2
\big|G_{\widetilde{Y}_-, \ell_{k+1}}\big|^{ \frac{\lambda_n}{\lambda_{\ell_{k+1}}}(1-\delta')} \frac{\delta}{100} \left(\frac{\lambda_n}{\lambda_{\ell_{k+1}}} \right)^2 
\e^{\frac{\lambda_n}{\lambda_{\ell_{k+1}}}H\big(\frac{\delta}{100}\big)}.
\]
By taking logarithm and dividing out by $\lambda_n$ we get
\begin{align*}
\frac{1}{\lambda_n}\log|E_{n, k}|
&\leq  H\big(C'2^{-2\ell_{k}}\big)
+ \frac{4}{\lambda_n}\log{\lambda_n} + (1-\delta')  \frac{1}{\lambda_{\ell_{k+1}}}\log \Big|G_{\widetilde{Y}_-, \ell_{k+1}}\Big|\\
& \quad + \frac{1}{\lambda_n}\log \frac{\delta}{100}+2\log\frac{1}{\lambda_{\ell_{k+1}}}+ \frac{1}{\lambda_{\ell_{k+1}}} H\left(\frac{\delta}{100}\right).
\end{align*}
Since $\delta\in [0,1]$, we have $\lambda_n^{-1} \log \delta(100)^{-1} <0$.
Since $\log x \leq x-1$ for every $x>0$, $H(t)\leq \log2\leq 1$ and $\lambda_n\geq9$ for $n\geq2$, we have
\[
2\log\frac{1}{\lambda_{\ell_{k+1}}}+ \frac{1}{\lambda_{\ell_{k+1}}} H\left(\frac{\delta}{100}\right)
\leq  2\left(\frac{1}{\lambda_{\ell_{k+1}}}-1\right)+\frac{1}{\lambda_{\ell_{k+1}}}
=\frac{3}{\lambda_{\ell_{k+1}}}-2\leq 0.
\]
For all $n$ large enough we have	
\[
\frac{4}{\lambda_n} \log \lambda_n \leq 2^{-2 \ell_{k}}.
\]
Hence we have						
\[
\frac{1}{\lambda_n}\log|E_{n, k}|
\leq H\big(C'2^{-2\ell_{k}}\big)+ (1-\delta') h_{k+1}^-+2^{-2\ell_{k}}
\]
for all $n$ large enough.
\end{proof}

Now we can finish the proof. 
We have a probability measure $\nu_{k+1}^-$ which is invariant ergodic under $\sigma_{\Lambda_{\ell_{k+1}}}$ and whose entropy is $h_{k+1}^-=\log|G_{\widetilde{Y}, \ell_{k+1}}|$ and 
whose support is included in $[G_{\widetilde{Y}, \ell_{k+1}}]$. (See Appendix \ref{appendix-pressure} for details.)
Since no forbidden pattern appears in $\Lambda_{\ell_{k+1}}$ for an element in $[G_{\widetilde{Y}, \ell_{k+1}}]$, we have
\[
\int S_{\ell_{k+1}}\phi \dd\nu_{k+1}^-
=\int_{[G_{\widetilde{Y}_-, \ell_{k+1}}]} \phi \dd\nu_{k+1}^-
\geq -\# (\Lambda_{\ell_{k+1}}\setminus \Lambda_{\ell_{k+1}-1})
=-8(\ell_{k+1}-1).
\]	
By the variational principle we have 
\begin{align*}
P(\beta_{k}\phi) & =P(\sigma, \beta_{k} \phi)
=\lambda_{\ell_{k+1}}^{-1} P(\sigma_{\Lambda_{\ell_{k+1}}}, \beta_{k} S_{\ell_{k+1}}\phi)\\
& \geq \frac{h(\nu_{k+1}^-)}{\lambda_{\ell_{k+1}}}+\frac{\beta_k}{\lambda_{\ell_{k+1}}} \int S_{\ell_{k+1}}\phi \dd \nu_{k+1}^-.
\end{align*}
The `block' shift $\sigma_{\Lambda_{\ell_k}}$ is defined in Appendix \ref{appendix-pressure}.
The second equality is a general fact. (We refer \cite[Theorem 9.8]{Wal82} for a proof in the case of a continuous transformation of a compact metric space and a continuous function $\phi$. The proof in the present context is obtained by combining the proof of Theorem 9.8 in \cite{Wal82} and Section 4.4 in \cite{Kel98}.)
Since $\mu_{\beta_{k}}$ is an equilibrium state for $\beta_{k}\phi$, we have
\begin{align}
h_{k+1}^-=\lambda_{\ell_{k+1}}^{-1}h(\nu_{k+1}^-)
&\leq P(\beta_{k}\phi)-\frac{\beta_k}{\lambda_{\ell_{k+1}}} \int S_{\ell_{k+1}}\phi \dd \nu_{k+1}^- \nonumber\\
&\leq h(\mu_{\beta_{k}}) +\beta_{k} \int \phi \dd\mu_{\beta_{k}}+\frac{\beta_k}{(2\ell_{k+1}-1)^2}\times 8(\ell_{k+1}-1) \nonumber\\
&\leq h(\mu_{\beta_{k}})+O(2^{-\ell_k}). 
\label{entropy inequation}
\end{align}
We get the last inequality because $ \phi\leq0$ and (S4).
The $L^1$ version of Shannon-McMillan-Breiman theorem \cite{Kel98} tells us that, for any $\eta>0$, there exists $n_0$
such that for any $n\geq n_0$ 
\[
\sum_{p\in \widetilde{\mathcal B}^{\Lambda_n}} \left| -\frac{1}{\lambda_n}\log\mu_{\beta_{k}}([p])-h(\mu_{\beta_{k}})\right| \mu_{\beta_{k}}([p])\leq \eta.
\]
In particular
\begin{align*}
\eta
&\geq \sum_{p\in [E_{n,k}]} \left| -\frac{1}{\lambda_n}\log\mu_{\beta_{k}}([p])-h(\mu_{\beta_{k}})\right| \mu_{\beta_{k}}([p])\\
& \geq \left|\sum_{p\in [E_{n,k}]} \Big( -\frac{1}{\lambda_n}\log\mu_{\beta_{k}}([p])-h(\mu_{\beta_{k}}) \Big)\mu_{\beta_{k}}([p])\right|
\end{align*}
from which it follows, by taking $\eta=2^{-{\ell_{k}}}$, that
\[
h(\mu_{\beta_{k}}) \mu_{\beta_{k}} ([E_{n,k}])
 \leq  \sum_{p\in [E_{n, k}]} \left(-\frac{1}{\lambda_n}\mu_{\beta_{k}}([p])\log \mu_{\beta_{k}}( [p]) \right) + 2^{-{\ell_{k}}}\,.
\]
Using Jensen inequality with the concave function $t\mapsto -t \log t$ and the weights $1/|[E_{n,k}]|$ we get
\[
h(\mu_{\beta_{k}}) \mu_{\beta_{k}} ([E_{n,k}])
\leq \frac{1}{\lambda_n}\log|[E_{n, k}]|-\frac{1}{\lambda_n} \mu_{\beta_{k}}([E_{n,k}] ) \log\mu_{\beta_{k}}([E_{n,k}] ) + 2^{-{\ell_{k}}}\,.
\]
We now apply Lemma \ref{estimate of Enk} with $\varepsilon=\frac{\delta'}{2}$ to get from \eqref{entropy inequation} that, for $n$ large enough,
\[
h_{k+1}^-\left(1-\frac{\delta'}{2}\right)
\leq \frac{1}{\lambda_n}\log|[E_{n, k}]|-\frac{1}{\lambda_n} \left(1-\frac{\delta'}{2}\right)\log\left(1-\frac{\delta'}{2}\right) + 2^{-{\ell_{k}}}+O(2^{-\ell_k}).
\]
Then Lemma \ref{estimation on exponential decay} implies
\[
h_{k+1}^-\left(1-\frac{\delta'}{2}\right)
\leq H\big(C'2^{-2\ell_{k}}\big)+ (1-\delta') h_{k+1}^-+o(h_{k+1}^-)+O(2^{-\ell_k}).
\]
Dividing out by $h_{k+1}^-$, we have
\begin{equation}
1-\frac{\delta'}{2} \leq 1-\delta'+ \frac{H\big(C'2^{-2\ell_{k}}\big)}{h_{k+1}^-}+\frac{O(2^{-\ell_k})}{h_{k+1}^-}+o(1).\label{final_estimate}
\end{equation}
Using \eqref{hkminus} and then letting $k\to \infty$ in \eqref{final_estimate}, we get a contradiction. 
\end{proof}

\appendix

\section{On partial sums of binomial coefficients}\label{appendix-binom}

Let $n\geq 1$ and $\alpha\in \left]0,1/2\right]$.
We have the following obvious lower bound:
\[
\sum_{r=0}^{\alpha n} \binom{n}{r}	 
\geq \sum_{r=0}^{\alpha n} \binom{\lfloor \alpha n\rfloor}{r}	
=2^{\lfloor \alpha n\rfloor}. 
\]
(Notice that we make a slight abuse of notation by writing $\alpha n$ instead of $\lfloor \alpha n \rfloor$.) \newline
Another bound we use is the following:
\[
\sum_{r=0}^{\alpha n} \binom{n}{r}\leq \e^{n H(\alpha)}
\]
where $H(\alpha)$ is defined in \eqref{def-binary-entropy}. 
For the reader's convenience, we give a proof since we could not find a handy reference.  
A straightforward calculation using the fact that
$\log \alpha-\log(1-\alpha)\leq 0$ for $\alpha\leq 1/2$ shows that for all $r\in \left[0, \alpha n\right]$
\[
r\log \alpha + (n-r)\log (1-\alpha) \geq -n H(\alpha)\,.
\]
Hence $\alpha^r (1-\alpha)^{n-r}\geq \e^{-n H(\alpha)}$, therefore
\[
1= (\alpha+(1-\alpha))^n
\geq \sum_{r=0}^{\alpha n} \binom{n}{r} \alpha^r (1-\alpha)^{n-r}
\geq \sum_{r=0}^{\alpha n} \binom{n}{r} \e^{-n H(\alpha)}
\]
which proves the desired bound.

\section{A measure with entropy $\log  |G_{\widetilde{Y}_{\pm}, \ell_{k+1}}|/\lambda_{\ell_{k+1}}$}\label{appendix-pressure}

Let $\mathcal{S}$ be an alphabet. Take $m\geq 1$ and consider the alphabet $\mathcal{T}=\mathcal{S}^{\Lambda_m}$.
Define  $\sigma_{\Lambda_m} : \mathcal{S}^{\Z^2}\rightarrow \mathcal{S}^{\Z^2}$ by 
\[
(\sigma_{\Lambda_m}^{\boldi}(x))_{\boldj}=x_{\boldi\Lambda_m+\boldj}=x_{(j_1+(2m+1)i_1, j_2+(2m+1)i_2)}
\]
for $\boldi =(i_1, i_2)\in \Z^2$ and $\boldj=(j_1, j_2)\in \Lambda_m$.
By $\boldi\Lambda_m$ we mean the dilated square $\{i_1(-m+1),\ldots,0,\ldots,i_1(m-1)\}\times \{i_2(-m+1),\ldots,0,\ldots,i_2(m-1)\}$.
The full shift over $\mathcal{T}$ is topologically conjugate with $(\mathcal{A}, \sigma_{\Lambda_m})$
and the conjugacies are the following:
$\mathfrak{f}: \mathcal{S}^{\Z^2}\rightarrow \mathcal{T}^{\Z^2}$ defined by 
$(\mathfrak{f}(x))_{\boldi}= (\sigma_{\Lambda_m}^{\boldi} x)_{\Lambda_m}$ and 
$\mathfrak{g}: \mathcal{T}^{\Z^2} \rightarrow \mathcal{S}^{\Z^2}$ defined by
$\left(\mathfrak{g}(y)\right)_{\boldi \Lambda_m+\boldj}=$ the alphabet at $\boldj$ in $\Lambda_m$ for $y_{\boldi}$.

Let $\mathcal{P}\subset \mathcal{T}$.
Then we have
\begin{align}
h_{{\scriptscriptstyle top}}(Z, \sigma_{\Lambda_m})=\log|\mathcal{P}|. 
\label{topological entropy}
\end{align}

%
Let $\tilde{\nu}$ be an ergodic measure on the full shift over $\mathcal{P}$ whose entropy is $\log|\mathcal{P}|$.
Let $\nu$ be the push forward of $\tilde{\nu}$ by $\mathfrak{g}$.
Note that $\nu$ is $\sigma_{\Lambda_m}$-invariant.

In Section \ref{Non-convergence} we considered 
$\mathcal{S}=\widetilde{\mathcal{B}}$, $m=\ell_k$ and $\mathcal{P}=G_{\widetilde{Y}_{-}, \ell_k}$
and the $\sigma_{\Lambda_{m}}$-invariant measure $\nu$ is $\nu_{k+1}^-$.

Since a generating $(n,1)$-separated set for $(X, \sigma_{\Lambda_m})$ is $((2m+1)n, 1)$-separating set for $(X, \sigma)$,
we have
\[
P(\sigma_{\Lambda_m}, S_m\phi)=\lambda_m P(\sigma, \phi)
\]
for every continuous function $\phi$ on $\mathcal{A}^{\Z^2}$. See \cite[Section 4.4]{Kel98} for details.
Since the support of $\tilde{\nu}$ is $\mathcal{P}^{\Z^2}$ and $\mathfrak{g}(\mathcal{P}^{\Z^2})\subset [\mathcal{P}]$, 
the support of $\nu$ is included in $[\mathcal{P}]$.

\bibliographystyle{abbrv}
\bibliography{reference_nonconvergence.bib}

\begin{bibdiv}
\begin{biblist}

\bib{AS13}{article}{
      author={Aubrun, N.},
      author={Sablik, M.},
       title={Simulation of effective subshifts by two-dimensional subshifts of
  finite type},
        date={2013},
        ISSN={0167-8019},
     journal={Acta Appl. Math.},
      volume={126},
       pages={35\ndash 63},
         url={https://doi-org.kras1.lib.keio.ac.jp/10.1007/s10440-013-9808-5},
      review={\MR{3077943}},
}

\bib{Bow75}{book}{
      author={Bowen, R.},
       title={Equilibrium states and the ergodic theory of {A}nosov
  diffeomorphisms},
     edition={revised},
      series={Lecture Notes in Mathematics},
   publisher={Springer-Verlag, Berlin},
        date={2008},
      volume={470},
        ISBN={978-3-540-77605-5},
        note={With a preface by David Ruelle, Edited by Jean-Ren\'{e}
  Chazottes},
      review={\MR{2423393}},
}

\bib{Bre03}{article}{
      author={Br\'{e}mont, Julien},
       title={Gibbs measures at temperature zero},
        date={2003},
        ISSN={0951-7715},
     journal={Nonlinearity},
      volume={16},
      number={2},
       pages={419\ndash 426},
         url={https://doi-org.kras1.lib.keio.ac.jp/10.1088/0951-7715/16/2/303},
      review={\MR{1958608}},
}

\bib{CGU11}{article}{
      author={Chazottes, J.-R.},
      author={Gambaudo, J.-M.},
      author={Ugalde, E.},
       title={Zero-temperature limit of one-dimensional {G}ibbs states via
  renormalization: the case of locally constant potentials},
        date={2011},
        ISSN={0143-3857},
     journal={Ergodic Theory Dynam. Systems},
      volume={31},
      number={4},
       pages={1109\ndash 1161},
         url={https://doi-org.kras1.lib.keio.ac.jp/10.1017/S014338571000026X},
      review={\MR{2818689}},
}

\bib{CH10}{article}{
      author={Chazottes, J.-R.},
      author={Hochman, M.},
       title={On the zero-temperature limit of {G}ibbs states},
        date={2010},
        ISSN={0010-3616},
     journal={Comm. Math. Phys.},
      volume={297},
      number={1},
       pages={265\ndash 281},
         url={https://doi-org.kras1.lib.keio.ac.jp/10.1007/s00220-010-0997-8},
      review={\MR{2645753}},
}

\bib{CR15}{article}{
      author={Coronel, D.},
      author={Rivera-Letelier, J.},
       title={Sensitive dependence of {G}ibbs measures at low temperatures},
        date={2015},
        ISSN={0022-4715},
     journal={J. Stat. Phys.},
      volume={160},
      number={6},
       pages={1658\ndash 1683},
         url={https://doi-org.kras1.lib.keio.ac.jp/10.1007/s10955-015-1288-8},
      review={\MR{3382762}},
}

\bib{dobrushin-Ising3D}{article}{
      author={Dobrushin, R.~L.},
       title={Gibbs state describing coexistence of phases for a
  three-dimensional ising model},
        date={1973},
     journal={Theory of Probability \& Its Applications},
      volume={17},
      number={4},
       pages={582\ndash 600},
         url={https://doi.org/10.1137/1117073},
}

\bib{DRS12}{article}{
      author={Durand, B.},
      author={Romashchenko, A.},
      author={Shen, A.},
       title={Fixed-point tile sets and their applications},
        date={2012},
        ISSN={0022-0000},
     journal={J. Comput. System Sci.},
      volume={78},
      number={3},
       pages={731\ndash 764},
         url={https://doi-org.kras1.lib.keio.ac.jp/10.1016/j.jcss.2011.11.001},
      review={\MR{2900032}},
}

\bib{DB1985}{incollection}{
      author={Dobrushin, R.~L.},
      author={Shlosman, S.~B.},
       title={The problem of translation invariance of {G}ibbs states at low
  temperatures},
        date={1985},
   booktitle={Mathematical physics reviews, {V}ol. 5},
      series={Soviet Sci. Rev. Sect. C Math. Phys. Rev.},
      volume={5},
   publisher={Harwood Academic Publ., Chur},
       pages={53\ndash 195},
      review={\MR{852217}},
}

\bib{Georgii}{book}{
      author={Georgii, H.-O.},
       title={Gibbs measures and phase transitions},
     edition={Second},
      series={De Gruyter Studies in Mathematics},
   publisher={Walter de Gruyter \& Co., Berlin},
        date={2011},
      volume={9},
        ISBN={978-3-11-025029-9},
         url={https://doi.org/10.1515/9783110250329},
      review={\MR{2807681}},
}

\bib{GT2015}{article}{
      author={Garibaldi, E.},
      author={Thieullen, P.},
       title={An ergodic description of ground states},
        date={2015},
        ISSN={0022-4715},
     journal={J. Stat. Phys.},
      volume={158},
      number={2},
       pages={359\ndash 371},
         url={https://doi.org/10.1007/s10955-014-1139-z},
      review={\MR{3299881}},
}

\bib{Hoc09}{article}{
      author={Hochman, M.},
       title={On the dynamics and recursive properties of multidimensional
  symbolic systems},
        date={2009},
        ISSN={0020-9910},
     journal={Invent. Math.},
      volume={176},
      number={1},
       pages={131\ndash 167},
         url={https://doi-org.kras1.lib.keio.ac.jp/10.1007/s00222-008-0161-7},
      review={\MR{2485881}},
}

\bib{Kel98}{book}{
      author={Keller, G.},
       title={Equilibrium states in ergodic theory},
      series={London Mathematical Society Student Texts},
   publisher={Cambridge University Press, Cambridge},
        date={1998},
      volume={42},
        ISBN={0-521-59420-0; 0-521-59534-7},
         url={https://doi-org.kras1.lib.keio.ac.jp/10.1017/CBO9781107359987},
      review={\MR{1618769}},
}

\bib{kurka-book}{book}{
      author={K{\.u}rka, P.},
       title={Topological and symbolic dynamics},
      series={Collection SMF: Cours sp{\'e}cialis{\'e}s},
   publisher={Soci{\'e}t{\'e} math{\'e}matique de France},
        date={2003},
        ISBN={9782856291436},
         url={https://books.google.fr/books?id=\_TMZAQAAIAAJ},
}

\bib{Lep05}{article}{
      author={Leplaideur, R.},
       title={A dynamical proof for the convergence of {G}ibbs measures at
  temperature zero},
        date={2005},
        ISSN={0951-7715},
     journal={Nonlinearity},
      volume={18},
      number={6},
       pages={2847\ndash 2880},
         url={https://doi-org.kras1.lib.keio.ac.jp/10.1088/0951-7715/18/6/023},
      review={\MR{2176962}},
}

\bib{LS02}{incollection}{
      author={Lind, D.},
      author={Schmidt, K.},
       title={Symbolic and algebraic dynamical systems},
        date={2002},
   booktitle={Handbook of dynamical systems, {V}ol. 1{A}},
   publisher={North-Holland, Amsterdam},
       pages={765\ndash 812},
  url={https://doi-org.kras1.lib.keio.ac.jp/10.1016/S1874-575X(02)80012-1},
      review={\MR{1928527}},
}

\bib{Ruelle2004}{book}{
      author={Ruelle, D.},
       title={Thermodynamic formalism},
     edition={Second},
      series={Cambridge Mathematical Library},
   publisher={Cambridge University Press, Cambridge},
        date={2004},
        ISBN={0-521-54649-4},
         url={https://doi.org/10.1017/CBO9780511617546},
        note={The mathematical structures of equilibrium statistical
  mechanics},
      review={\MR{2129258}},
}

\bib{vE-F-S}{article}{
      author={van Enter, A. C.~D.},
      author={Fern{\'a}ndez, R.},
      author={Sokal, A.~D.},
       title={Regularity properties and pathologies of position-space
  renormalization-group transformations: scope and limitations of {G}ibbsian
  theory},
        date={1993},
        ISSN={0022-4715},
     journal={J. Statist. Phys.},
      volume={72},
      number={5-6},
       pages={879\ndash 1167},
         url={http://dx.doi.org/10.1007/BF01048183},
      review={\MR{1241537 (94m:82012)}},
}

\bib{vE-Wioletta}{article}{
      author={van Enter, A. C.~D.},
      author={Ruszel, W.~M.},
       title={Chaotic temperature dependence at zero temperature},
        date={2007},
        ISSN={0022-4715},
     journal={J. Stat. Phys.},
      volume={127},
      number={3},
       pages={567\ndash 573},
         url={http://dx.doi.org/10.1007/s10955-006-9260-2},
      review={\MR{2316198 (2008f:82008)}},
}

\bib{Wal82}{book}{
      author={Walters, P.},
       title={An introduction to ergodic theory},
      series={Graduate Texts in Mathematics},
   publisher={Springer-Verlag, New York-Berlin},
        date={1982},
      volume={79},
        ISBN={0-387-90599-5},
      review={\MR{648108}},
}

\end{biblist}
\end{bibdiv}

\end{document}